\documentclass[aps,pra,amsmath,amssymb,reprint,superscriptaddress,showkeys,tightenlines]{revtex4-2}
\usepackage{array}
\usepackage{multirow}
\usepackage{dcolumn}
\usepackage{bm}
\usepackage[utf8]{inputenc}
\usepackage[T1]{fontenc}
\usepackage{lmodern}
\usepackage{hyperref}
\usepackage{graphicx}
\usepackage[dvipdfmx]{color}
\usepackage{float}
\usepackage{url}
\usepackage{braket}
\usepackage{here}
\usepackage{amsmath}
\usepackage{amssymb}
\usepackage{amsthm}

\newtheorem*{thm*}{Theorem}
\usepackage[FIGTOPCAP,hang,nooneline]{subfigure}

\hypersetup{
    unicode=false,          
    pdftoolbar=true,        
    pdfmenubar=true,        
    pdffitwindow=false,     
    pdfstartview={FitH},    
    pdfkeywords={keyword1} {key2} {key3}, 
    pdfnewwindow=true,      
    colorlinks=true,       
    linkcolor=red,          
    citecolor=blue,        
    filecolor=magenta,      
    urlcolor=cyan,           
}

\begin{document}
\theoremstyle{plain}
\title{Resource Reduction in Multiplexed High-Dimensional Quantum Reed-Solomon Codes}

\author{Shin Nishio}  \email{parton@nii.ac.jp}
  \affiliation{Department of Informatics, School of Multidisciplinary Sciences, SOKENDAI (The Graduate University for Advanced Studies), 2-1-2 Hitotsubashi, Chiyoda-ku, Tokyo, 101-8430, Japan}
  \affiliation{Quantum Information Science and Technology Unit, Okinawa Institute of Science and Technology Graduate University, Onna-son, Kunigami-gun, Okinawa 904-0495, Japan}
  \affiliation{National Institute of Informatics, 2-1-2 Hitotsubashi, Chiyoda-ku, Tokyo, 101-8430, Japan}

\author{Nicolò Lo Piparo}  
  \affiliation{Quantum Information Science and Technology Unit, Okinawa Institute of Science and Technology Graduate University, Onna-son, Kunigami-gun, Okinawa 904-0495, Japan}
  \affiliation{National Institute of Informatics, 2-1-2 Hitotsubashi, Chiyoda-ku, Tokyo, 101-8430, Japan}

\author{Michael Hanks} 
\affiliation{QOLS, Blackett Laboratory, Imperial College London, London SW7 2AZ, United Kingdom}

\author{William John Munro}
  \affiliation{NTT Basic Research Laboratories \& NTT Research Center for Theoretical Quantum Physics,
NTT Corporation, 3-1 Morinosato-Wakamiya, Atsugi, Kanagawa, 243-0198, Japan.}
  \affiliation{National Institute of Informatics, 2-1-2 Hitotsubashi, Chiyoda-ku, Tokyo, 101-8430, Japan}

\author{Kae Nemoto} \email{nemoto@nii.ac.jp}
  \affiliation{Quantum Information Science and Technology Unit, Okinawa Institute of Science and Technology Graduate University, Onna-son, Kunigami-gun, Okinawa 904-0495, Japan}
  \affiliation{National Institute of Informatics, 2-1-2 Hitotsubashi, Chiyoda-ku, Tokyo, 101-8430, Japan}
  \affiliation{Department of Informatics, School of Multidisciplinary Sciences, SOKENDAI (The Graduate University for Advanced Studies), 2-1-2 Hitotsubashi, Chiyoda-ku, Tokyo, 101-8430, Japan}  

\begin{abstract}
Quantum communication technologies will play an important role in quantum information processing in the near future as we network devices together. However, their implementation is still a challenging task due to both loss and gate errors. Quantum error correction codes are one important technique to address this issue. In particular, the Quantum Reed-Solomon codes are known to be quite efficient for quantum communication tasks. The high degree of physical resources required, however, makes such a code difficult to use in practice. A recent technique called quantum multiplexing has been shown to reduce resources  by using multiple degrees of freedom of a photon. In this work, we propose a method to decompose multi-controlled gates using fewer $\rm{CX}$ gates via this quantum multiplexing technique. We show that our method can significantly reduce the required number of $\rm{CX}$ gates needed in the encoding circuits for the quantum Reed-Solomon code. Our approach is also applicable to many other quantum error correction codes and quantum algorithms, including Grovers and quantum walks.
\end{abstract}
\maketitle
\section{Introduction}
Quantum communication systems utilizing the principles of superposition and entanglement will provide new capabilities that we can not achieve in our current telecommunication counterparts \cite{munro2015inside, duan2001long, jiang2009quantum, lo1999unconditional, hwang2003quantum}. Such communication is expected to provide the basis for distributed quantum information processing systems, including quantum key distribution \cite{bennett2020quantum, ekert1991quantum}, blind quantum computation \cite{arrighi2006blind, broadbent2009universal}, distributed quantum computation \cite{cirac1999distributed}, quantum remote sensing \cite{schanda2012physical}, and the quantum internet \cite{kimble2008quantum}. However, the fragile nature of quantum states as they propagate through such channels will severely limit the performance of the systems. Quantum error correction codes (QECCs) are a fundamental tool to address these issues as they can correct both channel loss and general errors (gate errors, for instance). 

Since the first QECC was found by Shor \cite{shor1995scheme}, quite a number of codes have been proposed, including the stabilizer codes \cite{gottesman1997stabilizer}, quantum low-density parity-check codes \cite{mackay2004sparse}, and GKP codes \cite{gottesman2001encoding}. Further, the Calderbank-Shor-Steane (CSS) codes \cite{calderbank1996good}  are widely studied because they are quantum codes that can be constructed using a variety of classical codes and inherit the good properties of the existing classical codes \cite{calderbank1996good}. For quantum computation, topological CSS codes, including the surface codes \cite{kitaev1997quantum} have known useful properties such as high thresholds and capability of practical transversal gates due to the locality of stabilizers. For communication and memory, the delay introduced in the decoding is not as important as it is in computation \cite{pierce1965failure}, so complex decoding circuits are acceptable compared to those typically used in the computation. Therefore, for quantum communication and quantum memory-related tasks, code rate and minimum distance are more important than the locality of the stabilizers. It is known that the Quantum Reed-Solomon code \cite{Grassl1999} is one of the most efficient codes for quantum communication from that point of view. 

The classical Reed-Solomon (RS) code \cite{reed1960polynomial} has been used in various communication systems \cite{wicker1999reed, feng2006protecting} due to it being a maximum distance separable (MDS) code \cite{macwilliams1977theory}. MDS codes are important because they can detect and correct the greatest number of errors for a fixed codeword length $n$ and message length $k$. In the RS code, multiple consecutive bits correspond to an element of a Galois Field (GF) \cite{reed1960polynomial}. This allows this code to be particularly resilient to a class of errors called burst errors (errors occurring in consecutive bits) \cite{wicker1999reed}, where multiple bits are used to represent a single element of the GF. Now the Quantum Reed-Solomon (QRS) codes \cite{Grassl1999} are known to be one of the most efficient CSS codes because of their excellent ability to correct qudit loss errors, which are the most critical issue in long-distance quantum communication in optics \cite{muralidharan2014ultrafast}. However, QECCs require significant physical sources (both in terms of the numbers of photons and qubits) for their realization, making their implementation quite challenging. Recently \cite{piparo2020resource} it was shown that quantum multiplexing allows one to reduce the number of resources when applied to the redundancy code \cite{ralph2005loss} and small scale QRS codes. This was achieved by using multiple degrees of freedom of the photon. Here we show that by using such quantum multiplexing techniques, we can drastically reduce the number of controlled-X ($\rm{CX}$) gates required to implement the encoding circuit of QRS codes that could be used in quantum communication tasks.

Now our paper is organized as follows. In Section \ref{sec_QRS}, we will present the construction of QRS codes and their coding circuits. We then briefly review in Section \ref{sec_QM} the quantum multiplexing technique and propose a method to implement multiple controlled gates ($C_{k}X$ gates) using fewer $\rm{CX}$ gates by utilizing quantum multiplexing. Section \ref{sec_result} then shows the degree of gate reduction that can be achieved when quantum multiplexing is applied to the encoding circuit of the QRS code. The results and applications of the proposed method to other quantum tasks are discussed in Section \ref{sec_conc}.

\section{The Quantum Reed-Solomon Code} \label{sec_QRS}
\begin{figure}[b]
\begin{center}
    \includegraphics[width=9cm, keepaspectratio]{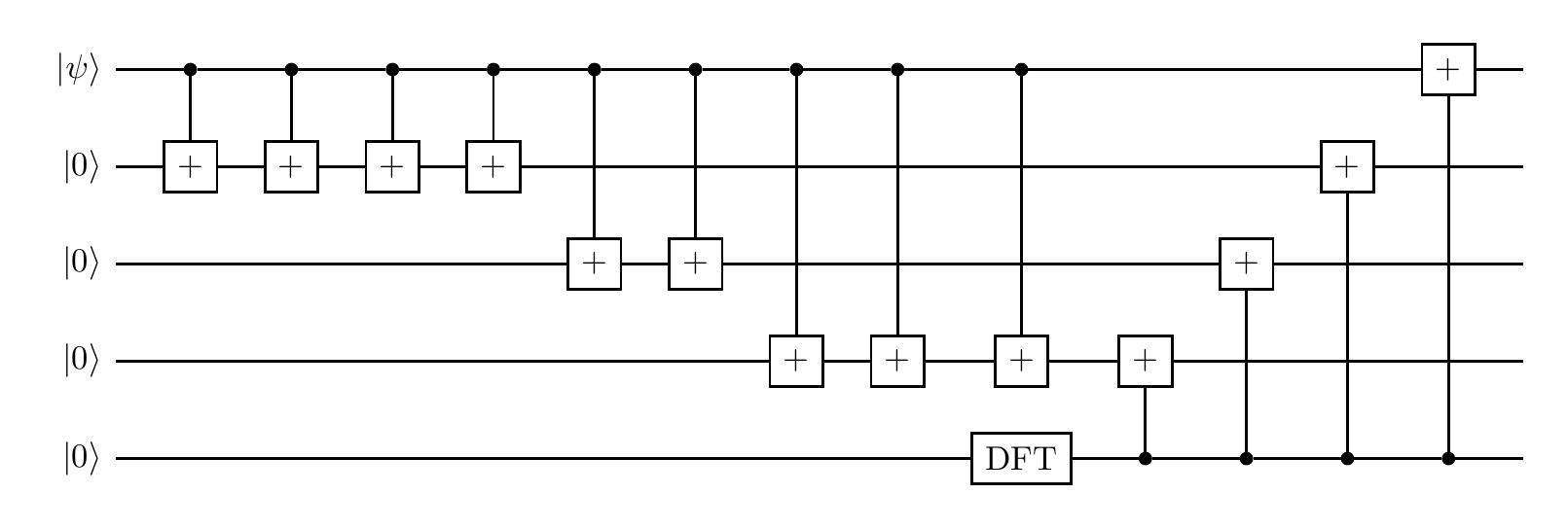}
    \caption{The encoding circuit for the $[[5,1,3]]_5$ Quantum Reed-Solomon code using $5$ dimensional qudits. The two qudit gates are defined by the black dot representing the control and the target with a square box containing the ``+'' symbol represent the SUM gates. The discrete Fourier transform gate is represented by the DFT labeled box.}
    \label{encoder}
\end{center}
\end{figure}
Let us begin with a brief overview of the Quantum Reed-Solomon code. The $[[d,2K-d,d-K+1]]_d$ QRS code \cite{Grassl1999, ketkar2006nonbinary, la2009constructions, muralidharan2018one} is defined by the CSS construction \cite{calderbank1996good} of the classical RS code in which $d$ (a prime number) is the dimension of the qudits used to encode the logic states. Further, $2K-d$ is the number of logical qudits, while $d-K+1$ is the minimum distance of the code. In our work, we will consider the $2K-d = 1$ case, which has the ability to retrieve the original encoded quantum information when $(d-1)/2$ or less qudits are lost. The encoding quantum circuit for the QRS code shown in Fig.\ref{encoder} can be implemented using a series of quantum SUM gates \cite{gottesman2001encoding} and a discrete Fourier transform (DFT) gate \cite{grassl2003efficient}. The SUM gate is a generalization of the $\rm{CX}$ gate for $d$-dimensional qudit and is given by:
$$\textrm{SUM}(\ket{A}\ket{B}) = \ket{A}\ket{(A+B) \,[\textrm{mod}\,d]}$$
where $A$ and $B$ are integers $\leq d-1$. Next the DFT gate is a generalization of the Hadamard gate, which when applied to a single qudit, creates a superposition of states given by:
$$\textrm{DFT}(\ket{0})=\frac{\ket{0}+\ket{1}+\cdots + \ket{d-1}}{\sqrt{d}}$$
To illustrate how these operations work, Fig.\ref{encoder} shows an example of the encoding circuit for the $[[5,1,3]]_5$ code.

The number of SUM gates required to create the $[[d,1,d+1/2]]_d$ code using a $d$-dimensional qudit is $(d^2 + d - 4)/2$, and is plotted in Fig.\ref{req_sum}. This number increases quadratically with the qudit dimension. Further, the number of $\rm{CX}$ gates used in each SUM gate also increases with the qudit dimension (as we will show later), making the implementation of higher-dimensional QRS codes more difficult. For our purposes, it is more convenient to encode each logical $d$-dimensional qudit of the QRS code with $k$ qubits where $2^{k-1}<d<2^k$. The SUM gate is essentially a modulo adder and can be decomposed into two parts. The first part is the Ripple carry adder (RCA) \cite{vedral1996quantum} that performs the following transformation:
$$\textrm{RCA}(\ket{A},\ket{B})=\ket{A}\ket{(A+B)\,[\textrm{mod} \,2^k]}.$$
The second part is a modulo operation that performs the following transformation:
$$\textrm{Mod}(\ket{A},\ket{(A+B)\,[\textrm{mod}\,2^k]}) =\ket{A}\ket{(A+B)\,[\textrm{mod}\,d]}.$$ 
In such a case, the RCA part adds two binary numbers having $k$ bits. This part utilizes auxiliary qubits called ``carry'' to store the outcome of the addition of two qubits having both values $1$. Further, the modulo operation can also be split into two distinct elements. In the first element, a series of logic gates check whether the outcome of the RCA exceeds $d$. In that case, the result will be stored in an auxiliary qubit we label ``check if''. A specific example of how these auxiliary qubits have been used is presented in Appendix \ref{modadder}. The second element, which we label as a conversion element, transforms the output of the RCA part in $[\textrm{mod}\,2^k]$ representation to the desired output of the SUM gate with $[\textrm{mod}\,d]$ representation.
\begin{figure}[h!]
\begin{center}
    \includegraphics[width=8cm, keepaspectratio]{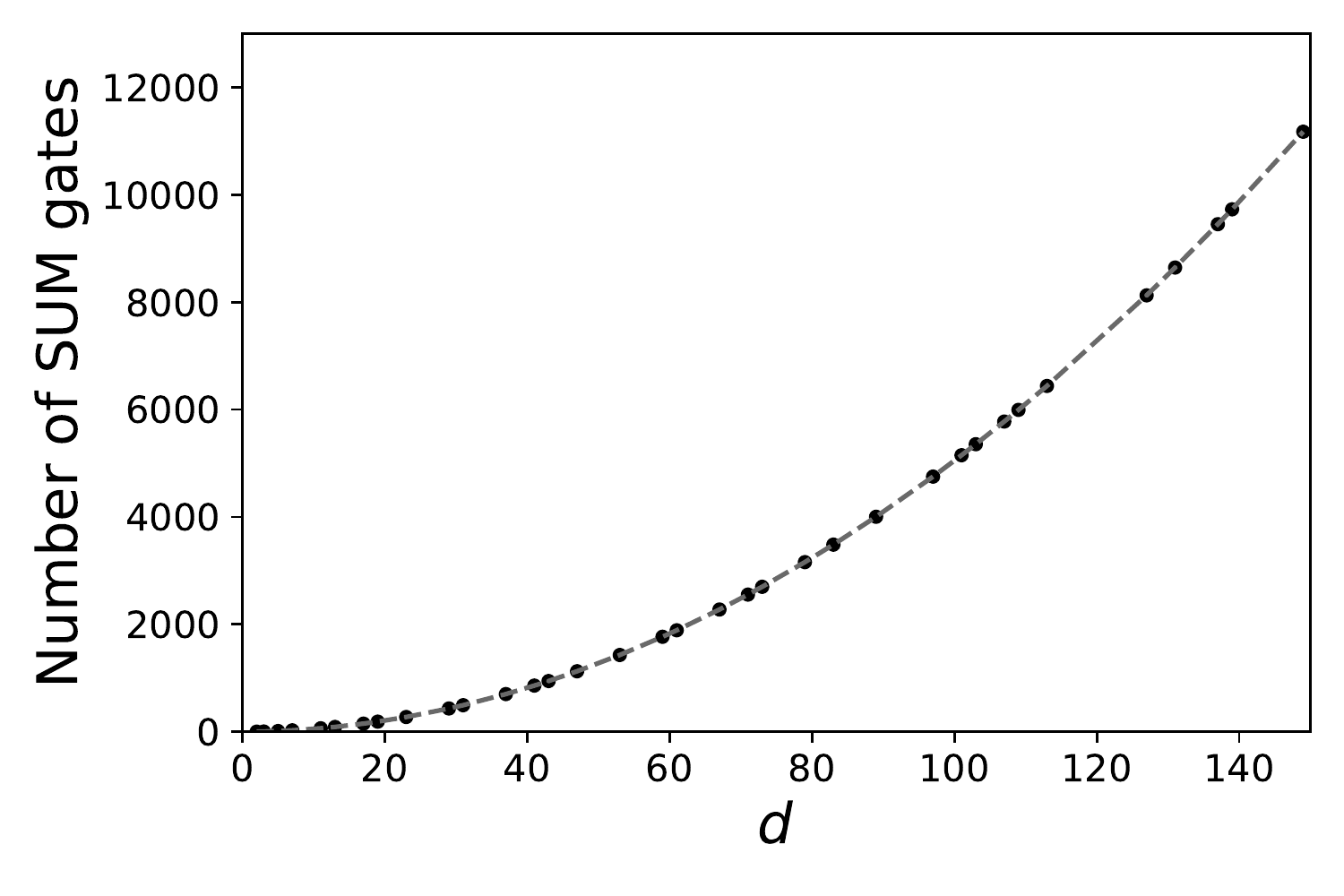}
    \caption{Required number of SUM gates for encoding the $[[d,1,(d+1)/2]]_d$ QRS codes versus qudit dimension $d$. The dots represent prime $d$ values. }
    \label{req_sum}
\end{center}
\end{figure}

We need to consider this in a little more detail and so let us look at the various outcomes of the RCA, which correspond to our three different cases. 
In the first case, where the outcome of the RCA is less than $d$, the conversion element will not change this outcome; hence no further action is needed.
In the second case, where the outcome is greater than $d-1$ but less than $2^k$, ``check if'' auxiliary qubits are used for storing each value. Those values are then used to convert the outcome to the desired state in the conversion element.
In the third case, where the outcome is greater than or equal to $2^k$, ``check if'' auxiliary qubits will be used to store each value, and the biggest ``carry'' qubit is also used to distinguish the value and the value minus $2^k$ in the conversion element. 
However, the ``check if'' qubit is not required when the outcome is $2^k$ if and only if $2(d-1)=2^k$. This is because the biggest ``carry'' can be used as the ``check if'' value since $2^k$ is the only one value that is bigger than $2^k$.
This SUM gate construction requires $k+d-2$ auxiliary qubits where $k$-qubits are for the ``carry'' and  $d-2$ qubits are for the ``check if''. 

The total number $N_{\rm{SUM}}$ of $\rm{CX}$ gates for each $\rm{SUM}$ gate in a $d$-qudit system with $2^{k-1} < d \leq 2^k$ is given by: $N_{\rm{RCA}} + N_{\rm{M}}$, where $N_{\rm{RCA}}(N_{\rm{M}})$ are the number of gates for the RCA (modulo) parts, respectively. The number of gates required by those are presented in the TABLE \ref{N_sum_table} where $C_{k}X$ is the $X$ gate controlled by $k$-qubits ($C_{2}X$ is controlled by $2$ and is the well known Toffoli gate), $C_{k+1}X$ gate is the $X$ gate controlled by $k$ qubits + 1 ``carry'' auxiliary qubit.

\begin{table*}[]
\begin{tabular}{l|l|l|l|l|l|}
\cline{2-6}
                                 & $C_{k+1}X$               & $C_kX$                 & $C_2 X$ & $C_1 X$ &     Condition                      \\ \hline
\multicolumn{1}{|l|}{$N_{\textrm{RCA}}$} &                           &                        & $3k-2$   & $2k-1$   &                           \\ \hline
\multicolumn{1}{|l|}{$N_M$}       &                           & $d-1$                    &        & $\displaystyle\sum_{i = d}^{2(d-1)} H_D(i,i\,[\textrm{mod}\,d])$       & $\rm{if}\,2(d-1)\leq 2^k$ \\ \cline{2-6} 
\multicolumn{1}{|l|}{}           & $2d-2^k-1$ & $2^k-d$ &        &$\displaystyle\sum_{i = d}^{2(d-1)} H_D(i,i\,[\textrm{mod}\,d])$        & $\rm{if}\,2(d-1)>2^k$     \\ \hline
\end{tabular}
\caption{The number of gates $N_{\textrm{RCA}}(N_M)$ required for the RCA (Modulo) operations in the $\rm{SUM}$ gate, respectively. $H_D(a, b)$ is a function that returns the Hamming distance between the binary representations of the input vector $a$ and $b$.}
\label{N_sum_table}
\end{table*}

Now in Fig.\ref{cx_full_adder} we plot $N_{\textrm{RCA}}$ versus $d$. Since the RCA part is a simple $k$ bit ripple carry adder, $N_{\textrm{RCA}}$ depends only on integer $k$ for $2^{k-1}<d<2^k$ which is shown as the gray line. When $d$ is sufficiently large, $N_{\textrm{RCA}}$ can be negligibly small compared to $N_M$, and so $N_{\textrm{SUM}}$ is dominated by $N_M$.

In the modulo conversion part, the $C_{k}X$ gate can be decomposed into $4(k-2)$ $\rm{C_{2}X}$ gates and the $C_{k+1}X$ gate can be decomposed into $4(k-1)$ $\rm{C_{2}X}$ gates as shown in \cite{barenco1995elementary}. Further the $\rm{C_{2}X}$ gate can be decomposed into $6\,\rm{CX}$, $2\,H$, $3\,T^\dagger$ and $5\,T$ gates. We will refer to this decomposition as a ``general decomposition'' and compare it to our more efficient multiplexing decomposition later on. In the next section, we will show that applying quantum multiplexing can drastically reduce the number of $\rm{CX}$ gates required to implement these circuits. The details of the modulo adder implementation are described in Appendix \ref{modadder}.

\setcounter{figure}{2}
\begin{figure}[htb]
    \begin{center}
        \includegraphics[width=8cm, keepaspectratio]{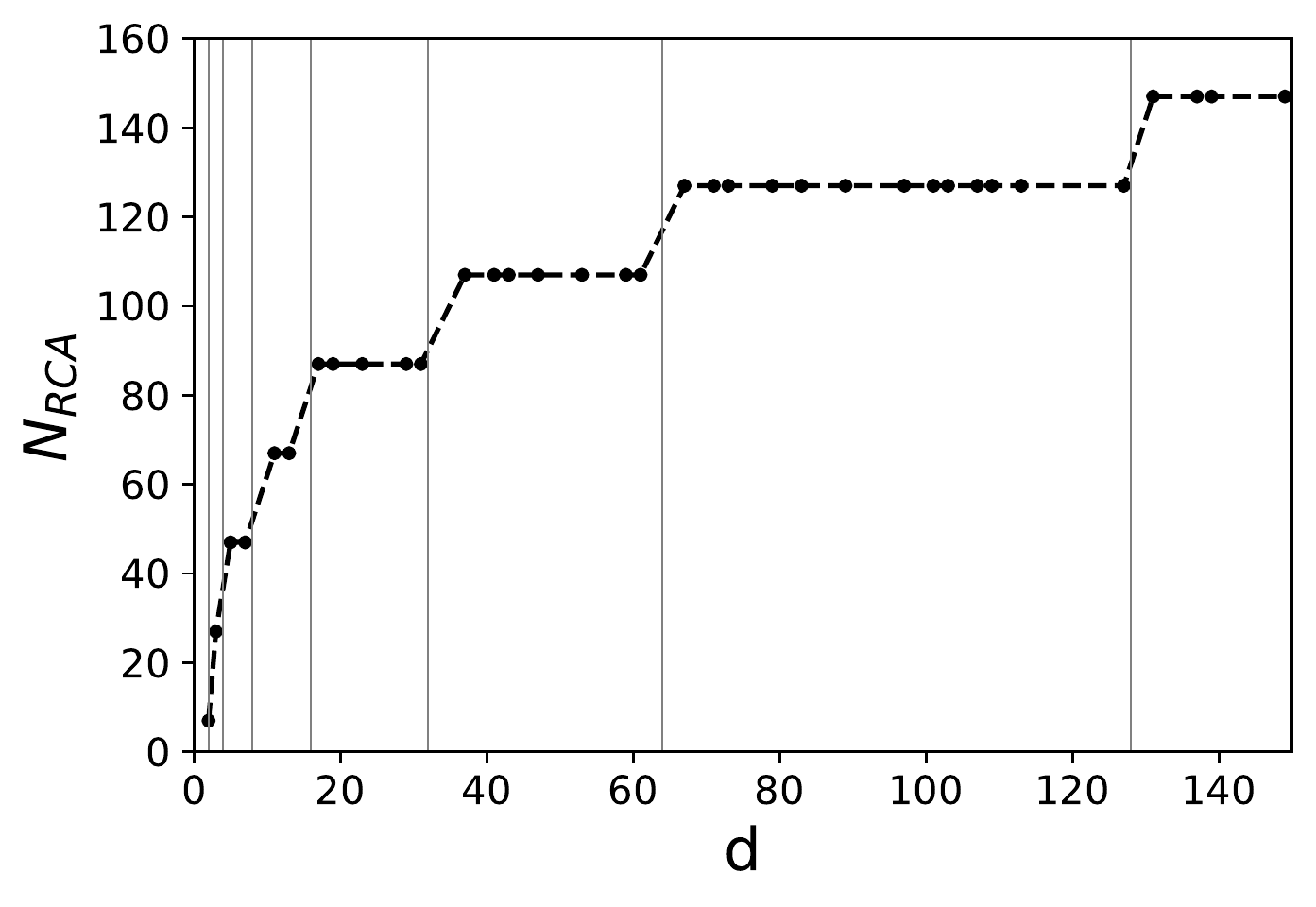}
        \caption{The number $N_{\rm{RCA}}$ of $\rm{CX}$ gates required for performing the RCA part of the SUM gate for $d$-dimensional qudits. The vertical gray lines correspond to the $2^m$ integer values, while the dotted lines between data points are used as a guide for the eye. }
        \label{cx_full_adder}
    \end{center}
\end{figure}

\section{Quantum Multiplexing}
\label{sec_QM}
The quantum multiplexing technique \cite{piparo2019quantum} aims at reducing the number of physical resources needed for a quantum communication system by using multiple degrees of freedom (DOF) of the single-photon to carry and process information. For instance, applying quantum multiplexing to error correction codes permits one to reduce drastically the number of single photons and qubits required to encode the information shared by Alice and Bob \cite{piparo2019quantum}. Another remarkable feature of quantum multiplexing is that it enables us to substitute a Toffoli gate \cite{piparo2020aggregating}, which normally requires $6\,\rm{CX}$ gates, with a single $\rm{CX}$ gate between two multiplexed photons (or no $\rm{CX}$ gates on a single photon). This is done by splitting one DOF of a photon into two separate spatial modes and applying a $\rm{CX}$ gate directly on one mode, as shown in Fig.\ref{multiplexing_and_circuit}(a). So the natural question is how we apply it to our situation. 
\setcounter{figure}{3} 
\begin{figure}[htb]
    \subfigure[]{
        \includegraphics[clip, width=1\columnwidth]{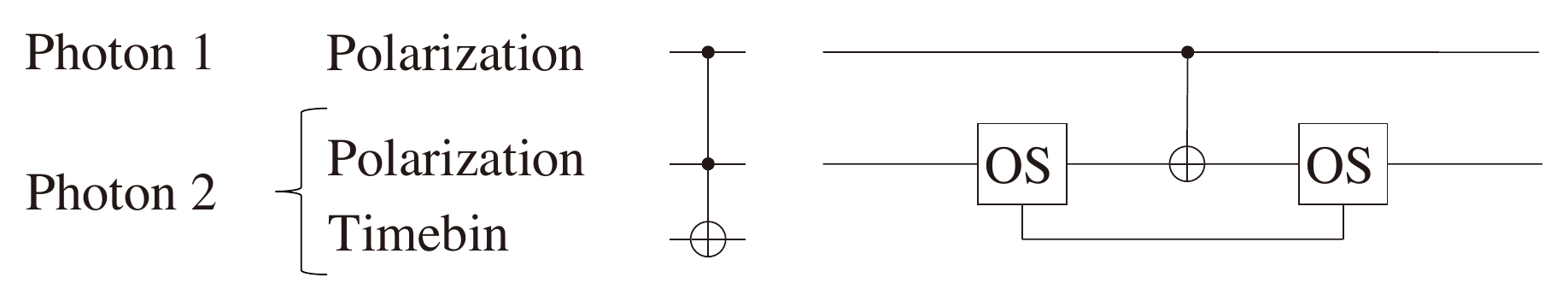}}
        \label{ccx_originals}
    \subfigure[]{
        \includegraphics[clip, width=1\columnwidth]{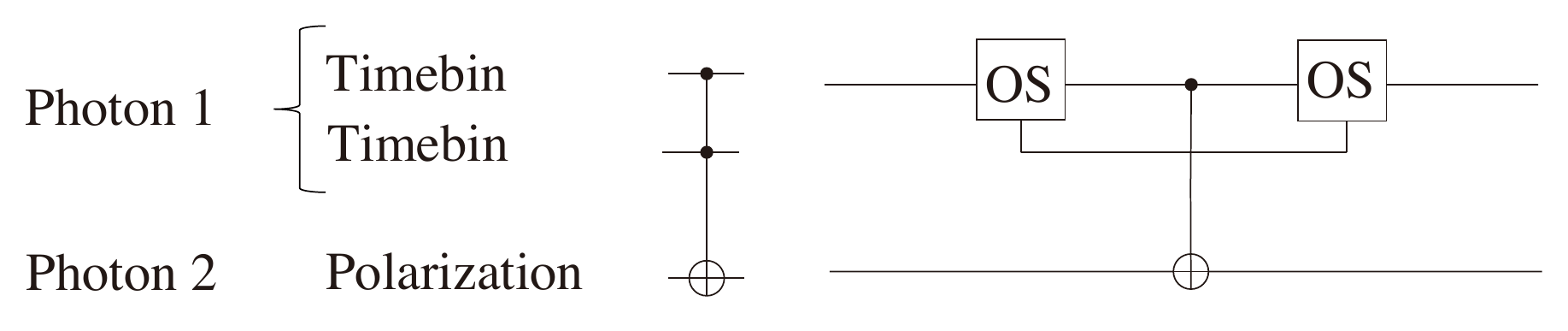}}%
        \label{ccx_dec}
    \subfigure[]{
        \includegraphics[clip, width=1\columnwidth]{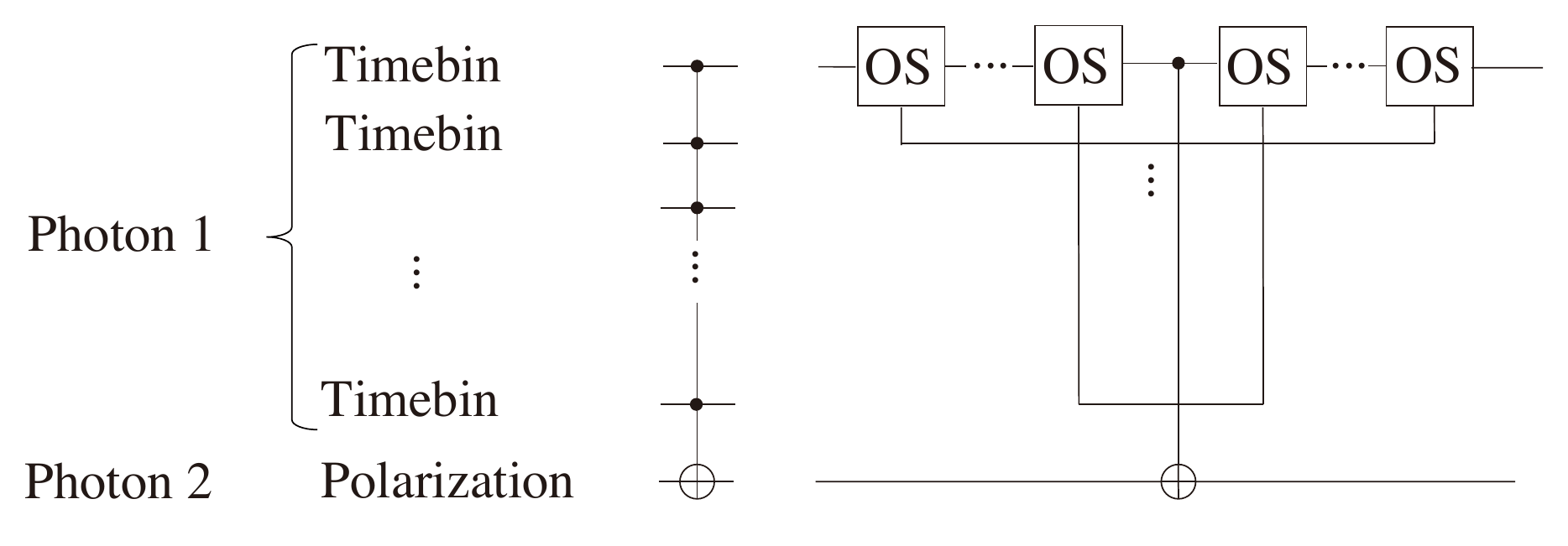}}%
        \label{ckx_dec}
    \caption{(a) Quantum circuit showing the $\rm{C_{2}X}$ gate between two photons (left) and a multiplexing circuit in which the $\rm{C_{2}X}$ gate is realized by an optical switch (OS) and a $\rm{CX}$ gate (right). (b) The $\rm{C_{2}X}$ gate a photon having as control two timebins and another photon having as target the polarization DOF (left) and its optical implementation (right). (c) The $C_{k}X$ gate between two photons (left) and its optical implementation (right).}
    \label{multiplexing_and_circuit}
\end{figure}

In our work, we first generalize this method to a $\rm{C_{2}X}$ gate between a photon having as control two timebins and another photon having as target the polarization DOF as shown in Fig.\ref{multiplexing_and_circuit}(b). Furthermore, we generalize the result of \cite{piparo2020aggregating} by assuming that we can also substitute the $C_{k}X$ gate and the $C_{k+1}X$ gate by a single $\rm{CX}$ and a $\rm{C_{2}X}$ gate respectively when quantum multiplexing is applied. To this aim, we require multiple optical switches OS (if the time-bin DOF is being used) to split each component belonging to the photon having multiple controls into single spatial modes. Then a single $\rm{CX}$ is applied between the relevant component and the target photon. Finally, a series of optical switches will restore the components into a single spatial mode, as shown in Fig.\ref{multiplexing_and_circuit}(c). Our assumption can be proved by using the induction principle shown in Appendix \ref{Proof}. 

We use this result then to substitute each $C_{k}X$ and $C_{k+1}X$ gates required in the modulo conversion part of our encoding scheme with a single $\rm{CX}$ gate and $\rm{C_{2}X}$ gate, respectively. Note that $C_{k+1}X$ requires a $\rm{C_{2}X}$ gate instead of a $\rm{CX}$ gate since $C_{k+1}X$ has a control qubit in the auxiliary photon. We will refer to this decomposition as a ``multiplexing decomposition''. This will considerably decrease the number of gates used, as we show in the next section.

\section{Performance}
\label{sec_result}
\begin{figure*}[ht]
    \subfigure[]{%
        \includegraphics[clip, width=1\columnwidth]{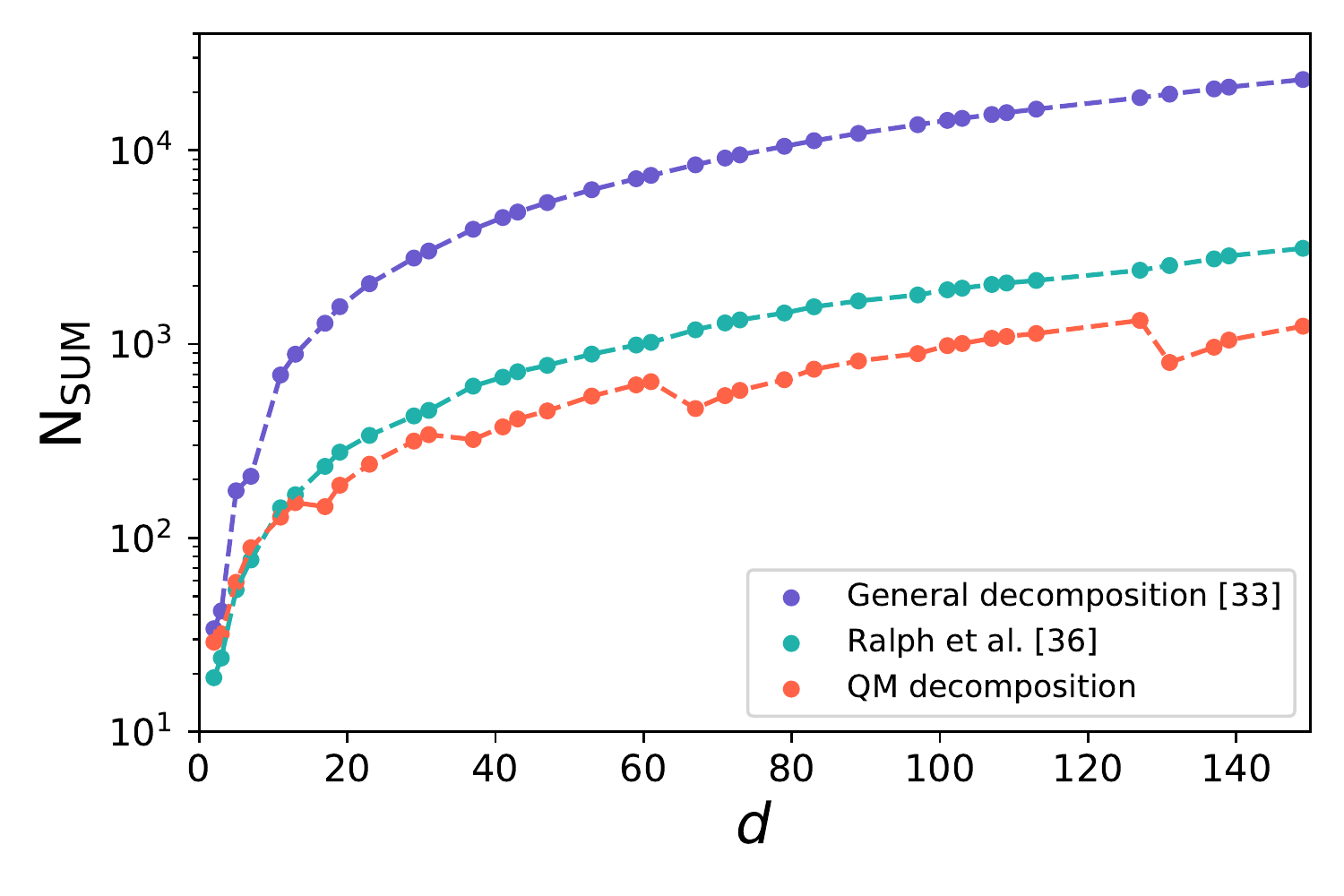}}%
        \label{sum_gate_cost}
    \subfigure[]{%
        \includegraphics[clip, width=1\columnwidth]{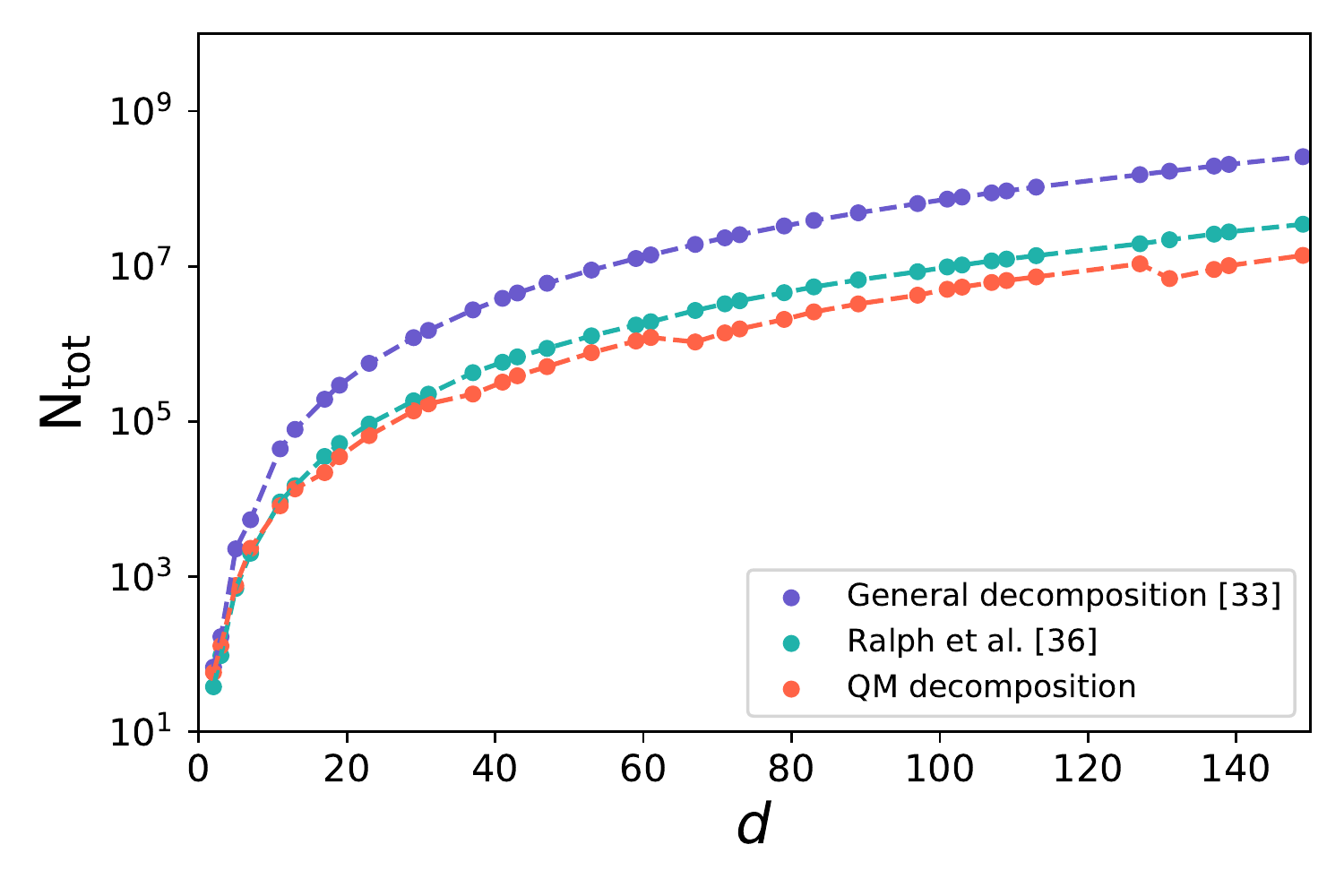}}%
        \label{encoder_cost}\\
    \subfigure[]{%
        \includegraphics[clip, width=1\columnwidth]{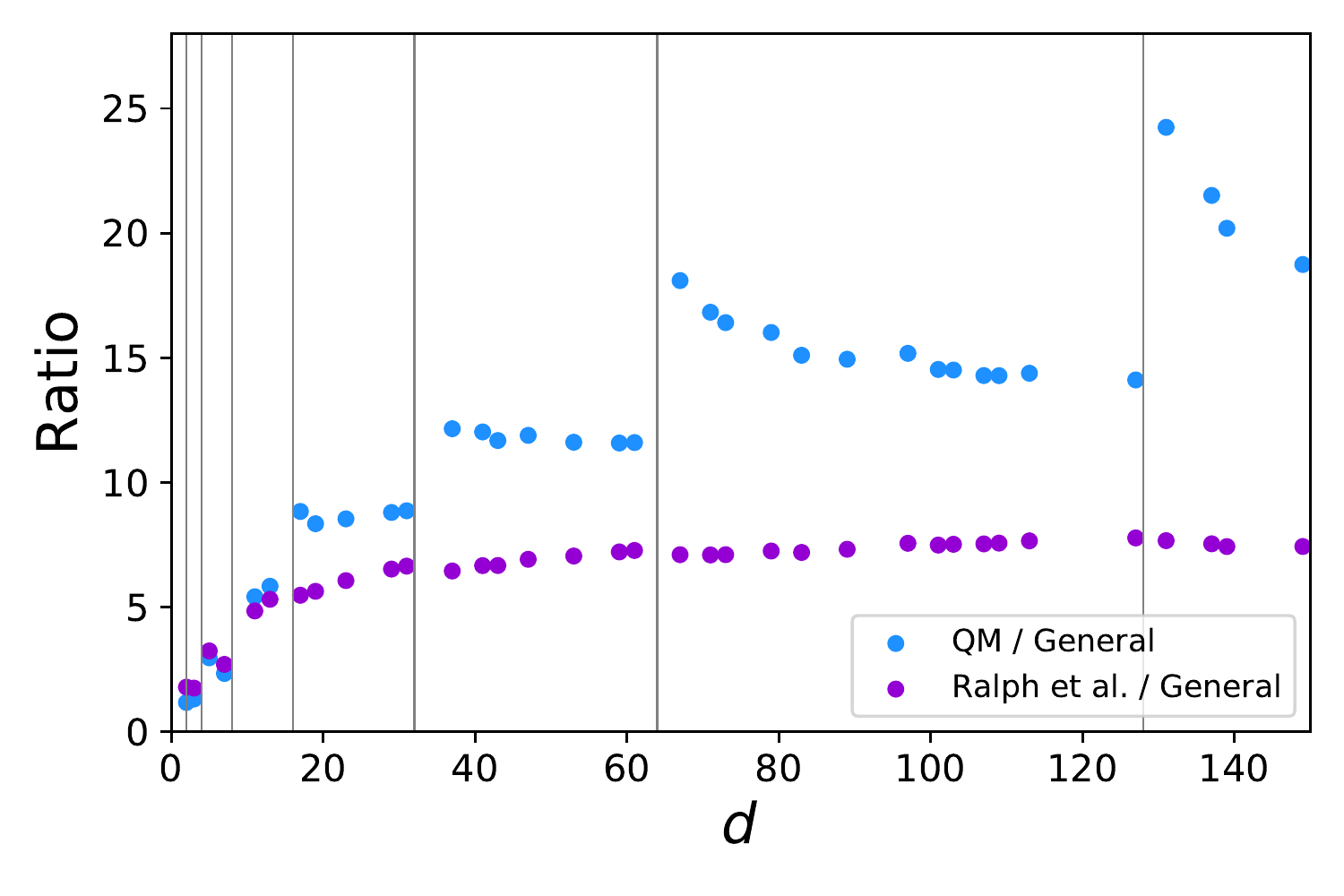}}%
        \label{sum_gate_ratio}
    \subfigure[]{%
        \includegraphics[clip, width=1\columnwidth]{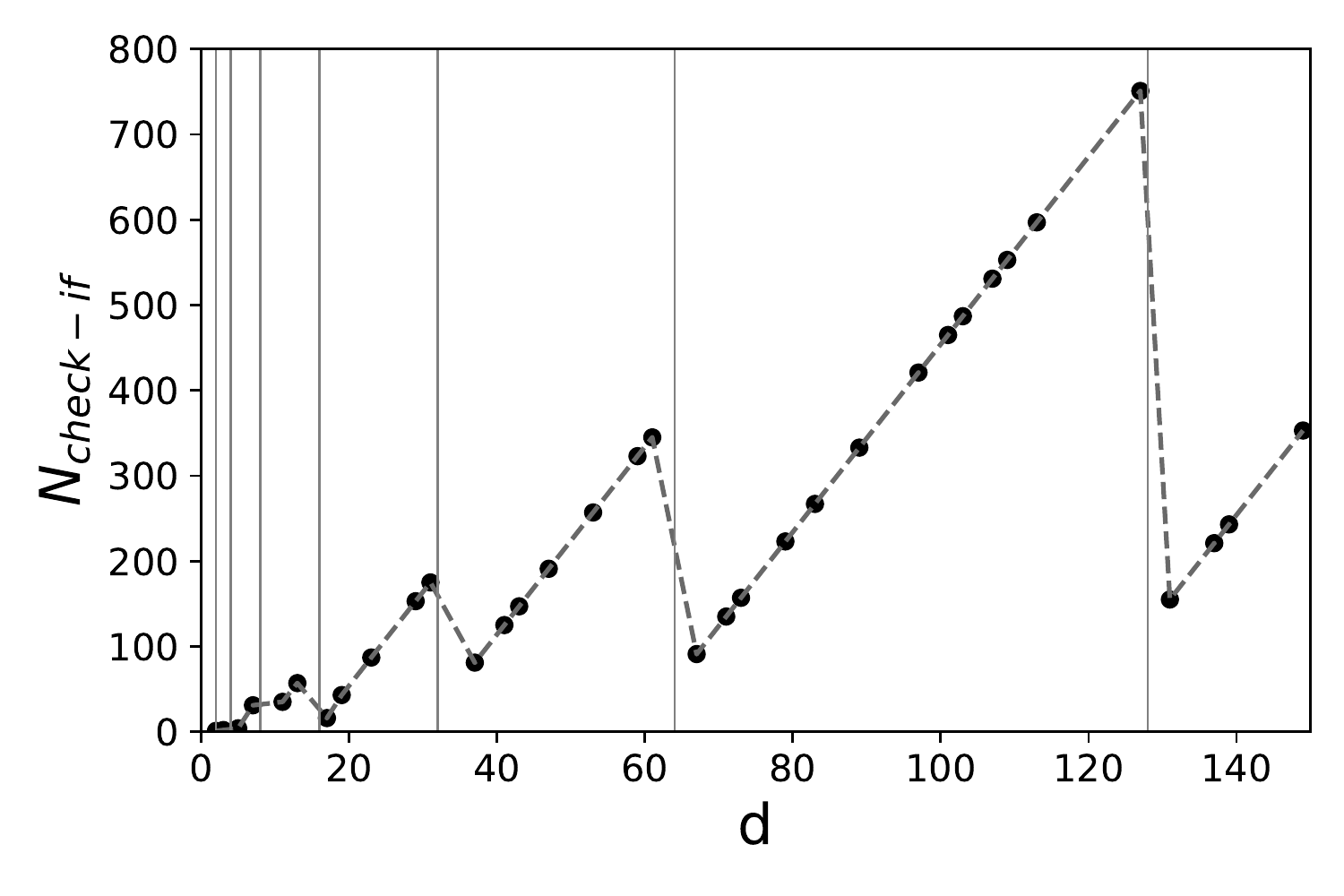}}
        \label{n_check-if}    
    \caption{(a) Plot of the required number of $\rm{CX}$ gates $N_{SUM}$ for the construction of a single SUM gate versus prime number $d$, the size of the qudit. The blue curve is based on the general decomposition \cite{barenco1995elementary} while the magenta curve is based on Ralph et al. \cite{ralph2007efficient} and the red curve uses the multiplexed decomposition shown in Fig.\ref{multiplexing_and_circuit}. (b) Plot of the total number of $\rm{CX}$ gates $N_{tot}$ for constructing the whole encoder of the $[[d,1,(d+1)/2]]_d$ QRS code versus $d$. (c) The ratio $R$ between the blue curve and the magenta (red) curve plotted in (a), respectively. The vertical gray lines correspond to $2^m$ for integer values of $m$. (d) Plot of the required number of $\rm{CX}$ gates $N_{check-if}$ for the ``check if'' part versus $d$ using quantum multiplexing.}
    \label{result}
\end{figure*}
Our focus now will be to compare the number of gates required to encode a QRS code versus its dimension when single-mode photons and multiplexed photons are used. The blue (red) curve of Fig.\ref{result}(a) shows the total number of $\rm{CX}$ gates required to implement a single SUM gate for a $d$-dimensional QRS code when single-mode (multiplexed) photons are used. The blue curve increases rapidly with the dimension of the code reaching $21182$ $\rm{CX}$ gates required for $d=139$ whereas the red curve increases modestly with the code dimension (using only $1049$ at $d=139$). Thus, applying quantum multiplexing has resulted in a significant reduction in CX gates. Moreover, for the multiplexing case, at $d=67$ and $d=137$ the number of $\rm{CX}$ gates is slightly less than $d=61$ and $d=131$, respectively. Although it seems an inconsistent result, it can be explained by looking at the second case of TABLE \ref{N_sum_table}. The number of $C_{k}X$ gates and the number of $C_{k+1}X$ gates are present in the ``check if'' part of the modulo conversion scheme. They reduce to $CX$ and $C_{2}X$ gates when multiplexing is applied. Therefore, the total number of $\rm{CX}$ gates will be much more affected by $N_{C_{k+1}X}$ than by $N_{C_{k}X}.$ These numbers depend on the number of qubits $k$ used to encode the states as well as on the dimension of the code, $d$. When $d\lesssim 2^{k}$, we require several ``check if'' auxiliary qubits to implement the modulo addition meaning $N_{C_{k+1}X}$ will be high. On the other hand, when $d\ll2^{k},$ $N_{C_{k}X}$ will be the preponderant term as we can now e use more qubits for the sum and less ``check if'' qubits. Although it is not clear from Fig.\ref{result}(a) we also have that $N_{C_{k+1}X}(d=31)>N_{C_{k+1}X}(d=37).$ This behavior is more evident in Fig.\ref{result}(b), which shows the total number of $\rm{CX}$ gates required to construct the whole encoder. The gate reduction by quantum multiplexing is also significant in the overall cost of the encoding circuit. Regardless of the value of $d$ we always require a number of conversion gates proportional to the code dimension. Thus, the number of conversion gates represented by the Hamming distance in TABLE \ref{N_sum_table} is not a relevant term for determining the behavior.

Next Fig.\ref{result}(c) shows the ratio of $N_{SUM}$ of the non-multiplexed case (blue) over the multiplexed case (red). It allows us to quantify the improvement we have. Here we can see that the curve ``jumps'' to a much higher ratio value when it crosses the gray lines, which correpond to $2^{k},$ with integer $k=2..7.$ We label the prime numbers in which this happens by $d_{\textrm{cross}}$. At $d>d_{\textrm{cross}}$ the ratio decreases as shown in regions ($d = 32$ to $63$, $64$ to $127$, and $128$ to $255$) of Fig.\ref{result}(c). This is due to the higher number of $\rm{C_{2}X}$ required as $d$ increases as already explained in the previous paragraph. However, at $d>19$ the ratio increases slightly until $d=31$ due to the fact that for this range of $d$ the number of ``check if'' increases only slightly as shows in Fig.\ref{result}(d). We also compare our results with the ones obtained by \cite{ralph2007efficient} shown by the magenta curves in Fig.\ref{result}. In \cite{ralph2007efficient}, the authors realize $C_k X$ gates by introducing one k-dimensional qudit, $2k-1$ two qubit gates and single qudit gates they refer to $X_a$ and $X_b$. Our results still show an advantage compared to the ones proposed in \cite{ralph2007efficient} in terms of gate reductions. Furthermore, the system proposed in \cite{ralph2007efficient} require special gates $X_a$ and $X_b$ which might require more two-qubit gates when qubits representation is in use.
We also determined that the encoding circuit for $[[d-1,d-2t,2t+1]]_d$ QRS codes where $d=2^m$ for integers $m$ does not require Toffoli or $C_{k}X$ gates, hence quantum multiplexing does not give any advantage in terms of gates reduction for these special cases. We show in the Appendix \ref{2m} an example of an implementation of a QRS encoding circuit over GF($2^m$). 

There is a well-known method \cite{cleve1997efficient} for constructing efficient encoding circuits for stabilizer codes, but this method requires multiple-target gates but not multiple-controlled gates. Therefore, the proposed method cannot be directly applied to such encoding circuits. The proposed method can be applied to the encoding circuit of quantum error correcting codes using multiple controlled gates.
 
\section{Discussion}
\label{sec_conc}
In this work, we have applied quantum multiplexing to the encoder of the Quantum Reed-Solomon and have shown that it is possible to dramatically reduce the number of controlled X gates utilized. Further, this method can be used to implement the encoder of QRS with fewer $\rm{CX}$ gates when the dimension of qudit $d$ is close to $2^k$ where $d$ is a prime number, and $k$ is greater than 6. This reduction of the number of gates should make the implementation of QRS codes more feasible - especially in the quantum communication setting. Another method of realizing $C_{k}X$ gates is to decompose the $C_{k}X$ gates directly into optical elements instead of $CX$ or $C_{2}X$ gates, which has been proposed in \cite{fiuravsek2006linear}. This method requires beam splitters with transmittance, mirrors, and phase shift for realizing multiple controlled-phase gates. For realizing $C_{k}X$ gates in real physical systems, it is necessary to quantitatively evaluate various resources and compare multiple implementation methods.

Next, our approach can be applied to other error correction codes and quantum algorithms that require a large number of gates. For example, most of the cost of a quantum circuit implementation of a discrete-time quantum walk algorithm \cite{aharonov1993quantum} is the unitary operator that walks a ``quantum walker'' in a certain space, called a Shift operator \cite{kempe2003quantum}. The circuit construction of the Shift operator depends on the space in which the quantum walk is performed (for instance, in a graph \cite{aharonov2001quantum}, one-dimensional space \cite{ambainis2001one}, or a hypercube \cite{moore2002quantum}), but in many cases, they consist of multiple controlled quantum gates \cite{douglas2009efficient}. This is because the Shift operator depends on the state of the ``quantum walker'' at a certain time and transitions to the quantum state of the ``quantum walker'' at the next time step, and such dependence on the previous time is realized by control gates. A second example is Grover's search algorithm \cite{grover1996fast}. The practical cost of Grover's algorithm depends on the complexity of the unitary, called an ``oracle'', which depends on the diffusion operator (inversion about the average). When the size of the search space is $N$-qubits, it is common to use $C_{N-1}X$ gates in the diffusion part \cite{lavor2003grover}.

To summarize, our paper shows that when quantum multiplexing is used, it is possible to significantly reduce the number of CX gates for the encoding circuits of QRS codes and important quantum algorithms. Future work includes optimization of circuits for correction and decoding of QRS codes and other QECCs.

\section*{Acknowledgements}
SN acknowledge Min-Hsiu Hsieh from Foxconn Quantum Computing Center for useful discusssions throughout this project. This work was supported by JSPS KAKENHI Grant Number 21H04880, the MEXT Quantum Leap Flagship Program (MEXT Q-LEAP) Grant Number JPMXS0118069605, and the JST Moonshot R\&D Grant Number JPMJMS2061.

\bibliography{main_bib}

\begin{thebibliography}{46}%
\makeatletter
\providecommand \@ifxundefined [1]{%
 \@ifx{#1\undefined}
}%
\providecommand \@ifnum [1]{%
 \ifnum #1\expandafter \@firstoftwo
 \else \expandafter \@secondoftwo
 \fi
}%
\providecommand \@ifx [1]{%
 \ifx #1\expandafter \@firstoftwo
 \else \expandafter \@secondoftwo
 \fi
}%
\providecommand \natexlab [1]{#1}%
\providecommand \enquote  [1]{``#1''}%
\providecommand \bibnamefont  [1]{#1}%
\providecommand \bibfnamefont [1]{#1}%
\providecommand \citenamefont [1]{#1}%
\providecommand \href@noop [0]{\@secondoftwo}%
\providecommand \href [0]{\begingroup \@sanitize@url \@href}%
\providecommand \@href[1]{\@@startlink{#1}\@@href}%
\providecommand \@@href[1]{\endgroup#1\@@endlink}%
\providecommand \@sanitize@url [0]{\catcode `\\12\catcode `\$12\catcode
  `\&12\catcode `\#12\catcode `\^12\catcode `\_12\catcode `\%12\relax}%
\providecommand \@@startlink[1]{}%
\providecommand \@@endlink[0]{}%
\providecommand \url  [0]{\begingroup\@sanitize@url \@url }%
\providecommand \@url [1]{\endgroup\@href {#1}{\urlprefix }}%
\providecommand \urlprefix  [0]{URL }%
\providecommand \Eprint [0]{\href }%
\providecommand \doibase [0]{https://doi.org/}%
\providecommand \selectlanguage [0]{\@gobble}%
\providecommand \bibinfo  [0]{\@secondoftwo}%
\providecommand \bibfield  [0]{\@secondoftwo}%
\providecommand \translation [1]{[#1]}%
\providecommand \BibitemOpen [0]{}%
\providecommand \bibitemStop [0]{}%
\providecommand \bibitemNoStop [0]{.\EOS\space}%
\providecommand \EOS [0]{\spacefactor3000\relax}%
\providecommand \BibitemShut  [1]{\csname bibitem#1\endcsname}%
\let\auto@bib@innerbib\@empty
\bibitem [{\citenamefont {Munro}\ \emph {et~al.}(2015)\citenamefont {Munro},
  \citenamefont {Azuma}, \citenamefont {Tamaki},\ and\ \citenamefont
  {Nemoto}}]{munro2015inside}%
  \BibitemOpen
  \bibfield  {author} {\bibinfo {author} {\bibfnamefont {W.~J.}\ \bibnamefont
  {Munro}}, \bibinfo {author} {\bibfnamefont {K.}~\bibnamefont {Azuma}},
  \bibinfo {author} {\bibfnamefont {K.}~\bibnamefont {Tamaki}},\ and\ \bibinfo
  {author} {\bibfnamefont {K.}~\bibnamefont {Nemoto}},\ }\bibfield  {title}
  {\bibinfo {title} {Inside quantum repeaters},\ }\href@noop {} {\bibfield
  {journal} {\bibinfo  {journal} {IEEE Journal of Selected Topics in Quantum
  Electronics}\ }\textbf {\bibinfo {volume} {21}},\ \bibinfo {pages} {78}
  (\bibinfo {year} {2015})}\BibitemShut {NoStop}%
\bibitem [{\citenamefont {Duan}\ \emph {et~al.}(2001)\citenamefont {Duan},
  \citenamefont {Lukin}, \citenamefont {Cirac},\ and\ \citenamefont
  {Zoller}}]{duan2001long}%
  \BibitemOpen
  \bibfield  {author} {\bibinfo {author} {\bibfnamefont {L.-M.}\ \bibnamefont
  {Duan}}, \bibinfo {author} {\bibfnamefont {M.~D.}\ \bibnamefont {Lukin}},
  \bibinfo {author} {\bibfnamefont {J.~I.}\ \bibnamefont {Cirac}},\ and\
  \bibinfo {author} {\bibfnamefont {P.}~\bibnamefont {Zoller}},\ }\bibfield
  {title} {\bibinfo {title} {Long-distance quantum communication with atomic
  ensembles and linear optics},\ }\href@noop {} {\bibfield  {journal} {\bibinfo
   {journal} {Nature}\ }\textbf {\bibinfo {volume} {414}},\ \bibinfo {pages}
  {413} (\bibinfo {year} {2001})}\BibitemShut {NoStop}%
\bibitem [{\citenamefont {Jiang}\ \emph {et~al.}(2009)\citenamefont {Jiang},
  \citenamefont {Taylor}, \citenamefont {Nemoto}, \citenamefont {Munro},
  \citenamefont {Van~Meter},\ and\ \citenamefont {Lukin}}]{jiang2009quantum}%
  \BibitemOpen
  \bibfield  {author} {\bibinfo {author} {\bibfnamefont {L.}~\bibnamefont
  {Jiang}}, \bibinfo {author} {\bibfnamefont {J.~M.}\ \bibnamefont {Taylor}},
  \bibinfo {author} {\bibfnamefont {K.}~\bibnamefont {Nemoto}}, \bibinfo
  {author} {\bibfnamefont {W.~J.}\ \bibnamefont {Munro}}, \bibinfo {author}
  {\bibfnamefont {R.}~\bibnamefont {Van~Meter}},\ and\ \bibinfo {author}
  {\bibfnamefont {M.~D.}\ \bibnamefont {Lukin}},\ }\bibfield  {title} {\bibinfo
  {title} {Quantum repeater with encoding},\ }\href@noop {} {\bibfield
  {journal} {\bibinfo  {journal} {Physical Review A}\ }\textbf {\bibinfo
  {volume} {79}},\ \bibinfo {pages} {032325} (\bibinfo {year}
  {2009})}\BibitemShut {NoStop}%
\bibitem [{\citenamefont {Lo}\ and\ \citenamefont
  {Chau}(1999)}]{lo1999unconditional}%
  \BibitemOpen
  \bibfield  {author} {\bibinfo {author} {\bibfnamefont {H.-K.}\ \bibnamefont
  {Lo}}\ and\ \bibinfo {author} {\bibfnamefont {H.~F.}\ \bibnamefont {Chau}},\
  }\bibfield  {title} {\bibinfo {title} {Unconditional security of quantum key
  distribution over arbitrarily long distances},\ }\href@noop {} {\bibfield
  {journal} {\bibinfo  {journal} {science}\ }\textbf {\bibinfo {volume}
  {283}},\ \bibinfo {pages} {2050} (\bibinfo {year} {1999})}\BibitemShut
  {NoStop}%
\bibitem [{\citenamefont {Hwang}(2003)}]{hwang2003quantum}%
  \BibitemOpen
  \bibfield  {author} {\bibinfo {author} {\bibfnamefont {W.-Y.}\ \bibnamefont
  {Hwang}},\ }\bibfield  {title} {\bibinfo {title} {Quantum key distribution
  with high loss: toward global secure communication},\ }\href@noop {}
  {\bibfield  {journal} {\bibinfo  {journal} {Physical Review Letters}\
  }\textbf {\bibinfo {volume} {91}},\ \bibinfo {pages} {057901} (\bibinfo
  {year} {2003})}\BibitemShut {NoStop}%
\bibitem [{\citenamefont {Bennett}\ and\ \citenamefont
  {Brassard}(1984)}]{bennett2020quantum}%
  \BibitemOpen
  \bibfield  {author} {\bibinfo {author} {\bibfnamefont {C.~H.}\ \bibnamefont
  {Bennett}}\ and\ \bibinfo {author} {\bibfnamefont {G.}~\bibnamefont
  {Brassard}},\ }\bibfield  {title} {\bibinfo {title} {Quantum cryptography:
  Public key distribution and coin tossing},\ }\href@noop {} {\bibfield
  {journal} {\bibinfo  {journal} {arXiv preprint arXiv:2003.06557}\ } (\bibinfo
  {year} {1984})}\BibitemShut {NoStop}%
\bibitem [{\citenamefont {Ekert}(1991)}]{ekert1991quantum}%
  \BibitemOpen
  \bibfield  {author} {\bibinfo {author} {\bibfnamefont {A.~K.}\ \bibnamefont
  {Ekert}},\ }\bibfield  {title} {\bibinfo {title} {Quantum cryptography based
  on bell’s theorem},\ }\href@noop {} {\bibfield  {journal} {\bibinfo
  {journal} {Physical Review Letters}\ }\textbf {\bibinfo {volume} {67}},\
  \bibinfo {pages} {661} (\bibinfo {year} {1991})}\BibitemShut {NoStop}%
\bibitem [{\citenamefont {Arrighi}\ and\ \citenamefont
  {Salvail}(2006)}]{arrighi2006blind}%
  \BibitemOpen
  \bibfield  {author} {\bibinfo {author} {\bibfnamefont {P.}~\bibnamefont
  {Arrighi}}\ and\ \bibinfo {author} {\bibfnamefont {L.}~\bibnamefont
  {Salvail}},\ }\bibfield  {title} {\bibinfo {title} {Blind quantum
  computation},\ }\href@noop {} {\bibfield  {journal} {\bibinfo  {journal}
  {International Journal of Quantum Information}\ }\textbf {\bibinfo {volume}
  {4}},\ \bibinfo {pages} {883} (\bibinfo {year} {2006})}\BibitemShut {NoStop}%
\bibitem [{\citenamefont {Broadbent}\ \emph {et~al.}(2009)\citenamefont
  {Broadbent}, \citenamefont {Fitzsimons},\ and\ \citenamefont
  {Kashefi}}]{broadbent2009universal}%
  \BibitemOpen
  \bibfield  {author} {\bibinfo {author} {\bibfnamefont {A.}~\bibnamefont
  {Broadbent}}, \bibinfo {author} {\bibfnamefont {J.}~\bibnamefont
  {Fitzsimons}},\ and\ \bibinfo {author} {\bibfnamefont {E.}~\bibnamefont
  {Kashefi}},\ }\bibfield  {title} {\bibinfo {title} {Universal blind quantum
  computation},\ }in\ \href@noop {} {\emph {\bibinfo {booktitle} {2009 50th
  Annual IEEE Symposium on Foundations of Computer Science}}}\ (\bibinfo
  {organization} {IEEE},\ \bibinfo {year} {2009})\ pp.\ \bibinfo {pages}
  {517--526}\BibitemShut {NoStop}%
\bibitem [{\citenamefont {Cirac}\ \emph {et~al.}(1999)\citenamefont {Cirac},
  \citenamefont {Ekert}, \citenamefont {Huelga},\ and\ \citenamefont
  {Macchiavello}}]{cirac1999distributed}%
  \BibitemOpen
  \bibfield  {author} {\bibinfo {author} {\bibfnamefont {J.}~\bibnamefont
  {Cirac}}, \bibinfo {author} {\bibfnamefont {A.}~\bibnamefont {Ekert}},
  \bibinfo {author} {\bibfnamefont {S.}~\bibnamefont {Huelga}},\ and\ \bibinfo
  {author} {\bibfnamefont {C.}~\bibnamefont {Macchiavello}},\ }\bibfield
  {title} {\bibinfo {title} {Distributed quantum computation over noisy
  channels},\ }\href@noop {} {\bibfield  {journal} {\bibinfo  {journal}
  {Physical Review A}\ }\textbf {\bibinfo {volume} {59}},\ \bibinfo {pages}
  {4249} (\bibinfo {year} {1999})}\BibitemShut {NoStop}%
\bibitem [{\citenamefont {Schanda}(2012)}]{schanda2012physical}%
  \BibitemOpen
  \bibfield  {author} {\bibinfo {author} {\bibfnamefont {E.}~\bibnamefont
  {Schanda}},\ }\href@noop {} {\emph {\bibinfo {title} {Physical fundamentals
  of remote sensing}}}\ (\bibinfo  {publisher} {Springer Science \& Business
  Media},\ \bibinfo {year} {2012})\BibitemShut {NoStop}%
\bibitem [{\citenamefont {Kimble}(2008)}]{kimble2008quantum}%
  \BibitemOpen
  \bibfield  {author} {\bibinfo {author} {\bibfnamefont {H.~J.}\ \bibnamefont
  {Kimble}},\ }\bibfield  {title} {\bibinfo {title} {The quantum internet},\
  }\href@noop {} {\bibfield  {journal} {\bibinfo  {journal} {Nature}\ }\textbf
  {\bibinfo {volume} {453}},\ \bibinfo {pages} {1023} (\bibinfo {year}
  {2008})}\BibitemShut {NoStop}%
\bibitem [{\citenamefont {Shor}(1995)}]{shor1995scheme}%
  \BibitemOpen
  \bibfield  {author} {\bibinfo {author} {\bibfnamefont {P.~W.}\ \bibnamefont
  {Shor}},\ }\bibfield  {title} {\bibinfo {title} {Scheme for reducing
  decoherence in quantum computer memory},\ }\href@noop {} {\bibfield
  {journal} {\bibinfo  {journal} {Physical Review A}\ }\textbf {\bibinfo
  {volume} {52}},\ \bibinfo {pages} {R2493} (\bibinfo {year}
  {1995})}\BibitemShut {NoStop}%
\bibitem [{\citenamefont {Gottesman}(1997)}]{gottesman1997stabilizer}%
  \BibitemOpen
  \bibfield  {author} {\bibinfo {author} {\bibfnamefont {D.}~\bibnamefont
  {Gottesman}},\ }\bibfield  {title} {\bibinfo {title} {Stabilizer codes and
  quantum error correction},\ }\href@noop {} {\bibfield  {journal} {\bibinfo
  {journal} {arXiv preprint quant-ph/9705052}\ } (\bibinfo {year}
  {1997})}\BibitemShut {NoStop}%
\bibitem [{\citenamefont {MacKay}\ \emph {et~al.}(2004)\citenamefont {MacKay},
  \citenamefont {Mitchison},\ and\ \citenamefont
  {McFadden}}]{mackay2004sparse}%
  \BibitemOpen
  \bibfield  {author} {\bibinfo {author} {\bibfnamefont {D.~J.}\ \bibnamefont
  {MacKay}}, \bibinfo {author} {\bibfnamefont {G.}~\bibnamefont {Mitchison}},\
  and\ \bibinfo {author} {\bibfnamefont {P.~L.}\ \bibnamefont {McFadden}},\
  }\bibfield  {title} {\bibinfo {title} {Sparse-graph codes for quantum error
  correction},\ }\href@noop {} {\bibfield  {journal} {\bibinfo  {journal} {IEEE
  Transactions on Information Theory}\ }\textbf {\bibinfo {volume} {50}},\
  \bibinfo {pages} {2315} (\bibinfo {year} {2004})}\BibitemShut {NoStop}%
\bibitem [{\citenamefont {Gottesman}\ \emph {et~al.}(2001)\citenamefont
  {Gottesman}, \citenamefont {Kitaev},\ and\ \citenamefont
  {Preskill}}]{gottesman2001encoding}%
  \BibitemOpen
  \bibfield  {author} {\bibinfo {author} {\bibfnamefont {D.}~\bibnamefont
  {Gottesman}}, \bibinfo {author} {\bibfnamefont {A.}~\bibnamefont {Kitaev}},\
  and\ \bibinfo {author} {\bibfnamefont {J.}~\bibnamefont {Preskill}},\
  }\bibfield  {title} {\bibinfo {title} {Encoding a qubit in an oscillator},\
  }\href@noop {} {\bibfield  {journal} {\bibinfo  {journal} {Physical Review
  A}\ }\textbf {\bibinfo {volume} {64}},\ \bibinfo {pages} {012310} (\bibinfo
  {year} {2001})}\BibitemShut {NoStop}%
\bibitem [{\citenamefont {Calderbank}\ and\ \citenamefont
  {Shor}(1996)}]{calderbank1996good}%
  \BibitemOpen
  \bibfield  {author} {\bibinfo {author} {\bibfnamefont {A.~R.}\ \bibnamefont
  {Calderbank}}\ and\ \bibinfo {author} {\bibfnamefont {P.~W.}\ \bibnamefont
  {Shor}},\ }\bibfield  {title} {\bibinfo {title} {Good quantum
  error-correcting codes exist},\ }\href@noop {} {\bibfield  {journal}
  {\bibinfo  {journal} {Physical Review A}\ }\textbf {\bibinfo {volume} {54}},\
  \bibinfo {pages} {1098} (\bibinfo {year} {1996})}\BibitemShut {NoStop}%
\bibitem [{\citenamefont {Kitaev}(1997)}]{kitaev1997quantum}%
  \BibitemOpen
  \bibfield  {author} {\bibinfo {author} {\bibfnamefont {A.~Y.}\ \bibnamefont
  {Kitaev}},\ }\bibfield  {title} {\bibinfo {title} {Quantum error correction
  with imperfect gates},\ }in\ \href@noop {} {\emph {\bibinfo {booktitle}
  {Quantum communication, computing, and measurement}}}\ (\bibinfo  {publisher}
  {Springer},\ \bibinfo {year} {1997})\ pp.\ \bibinfo {pages}
  {181--188}\BibitemShut {NoStop}%
\bibitem [{\citenamefont {Pierce}(1965)}]{pierce1965failure}%
  \BibitemOpen
  \bibfield  {author} {\bibinfo {author} {\bibfnamefont {W.~H.}\ \bibnamefont
  {Pierce}},\ }\href@noop {} {\emph {\bibinfo {title} {Failure-tolerant
  computer design}}}\ (\bibinfo  {publisher} {Academic Press},\ \bibinfo {year}
  {1965})\BibitemShut {NoStop}%
\bibitem [{\citenamefont {Grassl}\ \emph {et~al.}(1999)\citenamefont {Grassl},
  \citenamefont {Geiselmann},\ and\ \citenamefont {Beth}}]{Grassl1999}%
  \BibitemOpen
  \bibfield  {author} {\bibinfo {author} {\bibfnamefont {M.}~\bibnamefont
  {Grassl}}, \bibinfo {author} {\bibfnamefont {W.}~\bibnamefont {Geiselmann}},\
  and\ \bibinfo {author} {\bibfnamefont {T.}~\bibnamefont {Beth}},\ }\bibfield
  {title} {\bibinfo {title} {Quantum reed---solomon codes},\ }in\ \href@noop {}
  {\emph {\bibinfo {booktitle} {Applied Algebra, Algebraic Algorithms and
  Error-Correcting Codes}}}\ (\bibinfo  {publisher} {Springer},\ \bibinfo
  {year} {1999})\ pp.\ \bibinfo {pages} {231--244}\BibitemShut {NoStop}%
\bibitem [{\citenamefont {Reed}\ and\ \citenamefont
  {Solomon}(1960)}]{reed1960polynomial}%
  \BibitemOpen
  \bibfield  {author} {\bibinfo {author} {\bibfnamefont {I.~S.}\ \bibnamefont
  {Reed}}\ and\ \bibinfo {author} {\bibfnamefont {G.}~\bibnamefont {Solomon}},\
  }\bibfield  {title} {\bibinfo {title} {Polynomial codes over certain finite
  fields},\ }\href@noop {} {\bibfield  {journal} {\bibinfo  {journal} {Journal
  of the Society for Industrial and Applied Mathematics}\ }\textbf {\bibinfo
  {volume} {8}},\ \bibinfo {pages} {300} (\bibinfo {year} {1960})}\BibitemShut
  {NoStop}%
\bibitem [{\citenamefont {Wicker}\ and\ \citenamefont
  {Bhargava}(1999)}]{wicker1999reed}%
  \BibitemOpen
  \bibfield  {author} {\bibinfo {author} {\bibfnamefont {S.~B.}\ \bibnamefont
  {Wicker}}\ and\ \bibinfo {author} {\bibfnamefont {V.~K.}\ \bibnamefont
  {Bhargava}},\ }\href@noop {} {\emph {\bibinfo {title} {Reed-Solomon codes and
  their applications}}}\ (\bibinfo  {publisher} {John Wiley \& Sons},\ \bibinfo
  {year} {1999})\BibitemShut {NoStop}%
\bibitem [{\citenamefont {Feng}\ and\ \citenamefont
  {Yuen}(2006)}]{feng2006protecting}%
  \BibitemOpen
  \bibfield  {author} {\bibinfo {author} {\bibfnamefont {Y.~C.}\ \bibnamefont
  {Feng}}\ and\ \bibinfo {author} {\bibfnamefont {P.~C.}\ \bibnamefont
  {Yuen}},\ }\bibfield  {title} {\bibinfo {title} {Protecting face biometric
  data on smartcard with reed-solomon code},\ }in\ \href@noop {} {\emph
  {\bibinfo {booktitle} {2006 Conference on Computer Vision and Pattern
  Recognition Workshop (CVPRW'06)}}}\ (\bibinfo {organization} {IEEE},\
  \bibinfo {year} {2006})\ pp.\ \bibinfo {pages} {29--29}\BibitemShut {NoStop}%
\bibitem [{\citenamefont {MacWilliams}\ and\ \citenamefont
  {Sloane}(1977)}]{macwilliams1977theory}%
  \BibitemOpen
  \bibfield  {author} {\bibinfo {author} {\bibfnamefont {F.~J.}\ \bibnamefont
  {MacWilliams}}\ and\ \bibinfo {author} {\bibfnamefont {N.~J.~A.}\
  \bibnamefont {Sloane}},\ }\href@noop {} {\emph {\bibinfo {title} {The theory
  of error correcting codes}}},\ Vol.~\bibinfo {volume} {16}\ (\bibinfo
  {publisher} {Elsevier},\ \bibinfo {year} {1977})\BibitemShut {NoStop}%
\bibitem [{\citenamefont {Muralidharan}\ \emph {et~al.}(2014)\citenamefont
  {Muralidharan}, \citenamefont {Kim}, \citenamefont {L{\"u}tkenhaus},
  \citenamefont {Lukin},\ and\ \citenamefont
  {Jiang}}]{muralidharan2014ultrafast}%
  \BibitemOpen
  \bibfield  {author} {\bibinfo {author} {\bibfnamefont {S.}~\bibnamefont
  {Muralidharan}}, \bibinfo {author} {\bibfnamefont {J.}~\bibnamefont {Kim}},
  \bibinfo {author} {\bibfnamefont {N.}~\bibnamefont {L{\"u}tkenhaus}},
  \bibinfo {author} {\bibfnamefont {M.~D.}\ \bibnamefont {Lukin}},\ and\
  \bibinfo {author} {\bibfnamefont {L.}~\bibnamefont {Jiang}},\ }\bibfield
  {title} {\bibinfo {title} {Ultrafast and fault-tolerant quantum communication
  across long distances},\ }\href@noop {} {\bibfield  {journal} {\bibinfo
  {journal} {Physical review letters}\ }\textbf {\bibinfo {volume} {112}},\
  \bibinfo {pages} {250501} (\bibinfo {year} {2014})}\BibitemShut {NoStop}%
\bibitem [{\citenamefont {Piparo}\ \emph
  {et~al.}(2020{\natexlab{a}})\citenamefont {Piparo}, \citenamefont {Hanks},
  \citenamefont {Gravel}, \citenamefont {Nemoto},\ and\ \citenamefont
  {Munro}}]{piparo2020resource}%
  \BibitemOpen
  \bibfield  {author} {\bibinfo {author} {\bibfnamefont {N.~L.}\ \bibnamefont
  {Piparo}}, \bibinfo {author} {\bibfnamefont {M.}~\bibnamefont {Hanks}},
  \bibinfo {author} {\bibfnamefont {C.}~\bibnamefont {Gravel}}, \bibinfo
  {author} {\bibfnamefont {K.}~\bibnamefont {Nemoto}},\ and\ \bibinfo {author}
  {\bibfnamefont {W.~J.}\ \bibnamefont {Munro}},\ }\bibfield  {title} {\bibinfo
  {title} {Resource reduction for distributed quantum information processing
  using quantum multiplexed photons},\ }\href@noop {} {\bibfield  {journal}
  {\bibinfo  {journal} {Physical Review Letters}\ }\textbf {\bibinfo {volume}
  {124}},\ \bibinfo {pages} {210503} (\bibinfo {year}
  {2020}{\natexlab{a}})}\BibitemShut {NoStop}%
\bibitem [{\citenamefont {Ralph}\ \emph {et~al.}(2005)\citenamefont {Ralph},
  \citenamefont {Hayes},\ and\ \citenamefont {Gilchrist}}]{ralph2005loss}%
  \BibitemOpen
  \bibfield  {author} {\bibinfo {author} {\bibfnamefont {T.~C.}\ \bibnamefont
  {Ralph}}, \bibinfo {author} {\bibfnamefont {A.}~\bibnamefont {Hayes}},\ and\
  \bibinfo {author} {\bibfnamefont {A.}~\bibnamefont {Gilchrist}},\ }\bibfield
  {title} {\bibinfo {title} {Loss-tolerant optical qubits},\ }\href@noop {}
  {\bibfield  {journal} {\bibinfo  {journal} {Physical Review Letters}\
  }\textbf {\bibinfo {volume} {95}},\ \bibinfo {pages} {100501} (\bibinfo
  {year} {2005})}\BibitemShut {NoStop}%
\bibitem [{\citenamefont {Ketkar}\ \emph {et~al.}(2006)\citenamefont {Ketkar},
  \citenamefont {Klappenecker}, \citenamefont {Kumar},\ and\ \citenamefont
  {Sarvepalli}}]{ketkar2006nonbinary}%
  \BibitemOpen
  \bibfield  {author} {\bibinfo {author} {\bibfnamefont {A.}~\bibnamefont
  {Ketkar}}, \bibinfo {author} {\bibfnamefont {A.}~\bibnamefont
  {Klappenecker}}, \bibinfo {author} {\bibfnamefont {S.}~\bibnamefont
  {Kumar}},\ and\ \bibinfo {author} {\bibfnamefont {P.~K.}\ \bibnamefont
  {Sarvepalli}},\ }\bibfield  {title} {\bibinfo {title} {Nonbinary stabilizer
  codes over finite fields},\ }\href@noop {} {\bibfield  {journal} {\bibinfo
  {journal} {IEEE Transactions on Information Theory}\ }\textbf {\bibinfo
  {volume} {52}},\ \bibinfo {pages} {4892} (\bibinfo {year}
  {2006})}\BibitemShut {NoStop}%
\bibitem [{\citenamefont {La~Guardia}(2009)}]{la2009constructions}%
  \BibitemOpen
  \bibfield  {author} {\bibinfo {author} {\bibfnamefont {G.~G.}\ \bibnamefont
  {La~Guardia}},\ }\bibfield  {title} {\bibinfo {title} {Constructions of new
  families of nonbinary quantum codes},\ }\href@noop {} {\bibfield  {journal}
  {\bibinfo  {journal} {Physical Review A}\ }\textbf {\bibinfo {volume} {80}},\
  \bibinfo {pages} {042331} (\bibinfo {year} {2009})}\BibitemShut {NoStop}%
\bibitem [{\citenamefont {Muralidharan}\ \emph {et~al.}(2018)\citenamefont
  {Muralidharan}, \citenamefont {Zou}, \citenamefont {Li},\ and\ \citenamefont
  {Jiang}}]{muralidharan2018one}%
  \BibitemOpen
  \bibfield  {author} {\bibinfo {author} {\bibfnamefont {S.}~\bibnamefont
  {Muralidharan}}, \bibinfo {author} {\bibfnamefont {C.-L.}\ \bibnamefont
  {Zou}}, \bibinfo {author} {\bibfnamefont {L.}~\bibnamefont {Li}},\ and\
  \bibinfo {author} {\bibfnamefont {L.}~\bibnamefont {Jiang}},\ }\bibfield
  {title} {\bibinfo {title} {One-way quantum repeaters with quantum
  reed-solomon codes},\ }\href@noop {} {\bibfield  {journal} {\bibinfo
  {journal} {Physical Review A}\ }\textbf {\bibinfo {volume} {97}},\ \bibinfo
  {pages} {052316} (\bibinfo {year} {2018})}\BibitemShut {NoStop}%
\bibitem [{\citenamefont {Grassl}\ \emph {et~al.}(2003)\citenamefont {Grassl},
  \citenamefont {R{\"o}tteler},\ and\ \citenamefont
  {Beth}}]{grassl2003efficient}%
  \BibitemOpen
  \bibfield  {author} {\bibinfo {author} {\bibfnamefont {M.}~\bibnamefont
  {Grassl}}, \bibinfo {author} {\bibfnamefont {M.}~\bibnamefont
  {R{\"o}tteler}},\ and\ \bibinfo {author} {\bibfnamefont {T.}~\bibnamefont
  {Beth}},\ }\bibfield  {title} {\bibinfo {title} {Efficient quantum circuits
  for non-qubit quantum error-correcting codes},\ }\href@noop {} {\bibfield
  {journal} {\bibinfo  {journal} {International Journal of Foundations of
  Computer Science}\ }\textbf {\bibinfo {volume} {14}},\ \bibinfo {pages} {757}
  (\bibinfo {year} {2003})}\BibitemShut {NoStop}%
\bibitem [{\citenamefont {Vedral}\ \emph {et~al.}(1996)\citenamefont {Vedral},
  \citenamefont {Barenco},\ and\ \citenamefont {Ekert}}]{vedral1996quantum}%
  \BibitemOpen
  \bibfield  {author} {\bibinfo {author} {\bibfnamefont {V.}~\bibnamefont
  {Vedral}}, \bibinfo {author} {\bibfnamefont {A.}~\bibnamefont {Barenco}},\
  and\ \bibinfo {author} {\bibfnamefont {A.}~\bibnamefont {Ekert}},\ }\bibfield
   {title} {\bibinfo {title} {Quantum networks for elementary arithmetic
  operations},\ }\href@noop {} {\bibfield  {journal} {\bibinfo  {journal}
  {Physical Review A}\ }\textbf {\bibinfo {volume} {54}},\ \bibinfo {pages}
  {147} (\bibinfo {year} {1996})}\BibitemShut {NoStop}%
\bibitem [{\citenamefont {Barenco}\ \emph {et~al.}(1995)\citenamefont
  {Barenco}, \citenamefont {Bennett}, \citenamefont {Cleve}, \citenamefont
  {DiVincenzo}, \citenamefont {Margolus}, \citenamefont {Shor}, \citenamefont
  {Sleator}, \citenamefont {Smolin},\ and\ \citenamefont
  {Weinfurter}}]{barenco1995elementary}%
  \BibitemOpen
  \bibfield  {author} {\bibinfo {author} {\bibfnamefont {A.}~\bibnamefont
  {Barenco}}, \bibinfo {author} {\bibfnamefont {C.~H.}\ \bibnamefont
  {Bennett}}, \bibinfo {author} {\bibfnamefont {R.}~\bibnamefont {Cleve}},
  \bibinfo {author} {\bibfnamefont {D.~P.}\ \bibnamefont {DiVincenzo}},
  \bibinfo {author} {\bibfnamefont {N.}~\bibnamefont {Margolus}}, \bibinfo
  {author} {\bibfnamefont {P.}~\bibnamefont {Shor}}, \bibinfo {author}
  {\bibfnamefont {T.}~\bibnamefont {Sleator}}, \bibinfo {author} {\bibfnamefont
  {J.~A.}\ \bibnamefont {Smolin}},\ and\ \bibinfo {author} {\bibfnamefont
  {H.}~\bibnamefont {Weinfurter}},\ }\bibfield  {title} {\bibinfo {title}
  {Elementary gates for quantum computation},\ }\href@noop {} {\bibfield
  {journal} {\bibinfo  {journal} {Physical Review A}\ }\textbf {\bibinfo
  {volume} {52}},\ \bibinfo {pages} {3457} (\bibinfo {year}
  {1995})}\BibitemShut {NoStop}%
\bibitem [{\citenamefont {Piparo}\ \emph {et~al.}(2019)\citenamefont {Piparo},
  \citenamefont {Munro},\ and\ \citenamefont {Nemoto}}]{piparo2019quantum}%
  \BibitemOpen
  \bibfield  {author} {\bibinfo {author} {\bibfnamefont {N.~L.}\ \bibnamefont
  {Piparo}}, \bibinfo {author} {\bibfnamefont {W.~J.}\ \bibnamefont {Munro}},\
  and\ \bibinfo {author} {\bibfnamefont {K.}~\bibnamefont {Nemoto}},\
  }\bibfield  {title} {\bibinfo {title} {Quantum multiplexing},\ }\href@noop {}
  {\bibfield  {journal} {\bibinfo  {journal} {Physical Review A}\ }\textbf
  {\bibinfo {volume} {99}},\ \bibinfo {pages} {022337} (\bibinfo {year}
  {2019})}\BibitemShut {NoStop}%
\bibitem [{\citenamefont {Piparo}\ \emph
  {et~al.}(2020{\natexlab{b}})\citenamefont {Piparo}, \citenamefont {Hanks},
  \citenamefont {Nemoto},\ and\ \citenamefont {Munro}}]{piparo2020aggregating}%
  \BibitemOpen
  \bibfield  {author} {\bibinfo {author} {\bibfnamefont {N.~L.}\ \bibnamefont
  {Piparo}}, \bibinfo {author} {\bibfnamefont {M.}~\bibnamefont {Hanks}},
  \bibinfo {author} {\bibfnamefont {K.}~\bibnamefont {Nemoto}},\ and\ \bibinfo
  {author} {\bibfnamefont {W.~J.}\ \bibnamefont {Munro}},\ }\bibfield  {title}
  {\bibinfo {title} {Aggregating quantum networks},\ }\href@noop {} {\bibfield
  {journal} {\bibinfo  {journal} {Physical Review A}\ }\textbf {\bibinfo
  {volume} {102}},\ \bibinfo {pages} {052613} (\bibinfo {year}
  {2020}{\natexlab{b}})}\BibitemShut {NoStop}%
\bibitem [{\citenamefont {Ralph}\ \emph {et~al.}(2007)\citenamefont {Ralph},
  \citenamefont {Resch},\ and\ \citenamefont {Gilchrist}}]{ralph2007efficient}%
  \BibitemOpen
  \bibfield  {author} {\bibinfo {author} {\bibfnamefont {T.}~\bibnamefont
  {Ralph}}, \bibinfo {author} {\bibfnamefont {K.}~\bibnamefont {Resch}},\ and\
  \bibinfo {author} {\bibfnamefont {A.}~\bibnamefont {Gilchrist}},\ }\bibfield
  {title} {\bibinfo {title} {Efficient toffoli gates using qudits},\
  }\href@noop {} {\bibfield  {journal} {\bibinfo  {journal} {Physical Review
  A}\ }\textbf {\bibinfo {volume} {75}},\ \bibinfo {pages} {022313} (\bibinfo
  {year} {2007})}\BibitemShut {NoStop}%
\bibitem [{\citenamefont {Cleve}\ and\ \citenamefont
  {Gottesman}(1997)}]{cleve1997efficient}%
  \BibitemOpen
  \bibfield  {author} {\bibinfo {author} {\bibfnamefont {R.}~\bibnamefont
  {Cleve}}\ and\ \bibinfo {author} {\bibfnamefont {D.}~\bibnamefont
  {Gottesman}},\ }\bibfield  {title} {\bibinfo {title} {Efficient computations
  of encodings for quantum error correction},\ }\href@noop {} {\bibfield
  {journal} {\bibinfo  {journal} {Physical Review A}\ }\textbf {\bibinfo
  {volume} {56}},\ \bibinfo {pages} {76} (\bibinfo {year} {1997})}\BibitemShut
  {NoStop}%
\bibitem [{\citenamefont {Fiur{\'a}{\v{s}}ek}(2006)}]{fiuravsek2006linear}%
  \BibitemOpen
  \bibfield  {author} {\bibinfo {author} {\bibfnamefont {J.}~\bibnamefont
  {Fiur{\'a}{\v{s}}ek}},\ }\bibfield  {title} {\bibinfo {title} {Linear-optics
  quantum toffoli and fredkin gates},\ }\href@noop {} {\bibfield  {journal}
  {\bibinfo  {journal} {Physical Review A}\ }\textbf {\bibinfo {volume} {73}},\
  \bibinfo {pages} {062313} (\bibinfo {year} {2006})}\BibitemShut {NoStop}%
\bibitem [{\citenamefont {Aharonov}\ \emph {et~al.}(1993)\citenamefont
  {Aharonov}, \citenamefont {Davidovich},\ and\ \citenamefont
  {Zagury}}]{aharonov1993quantum}%
  \BibitemOpen
  \bibfield  {author} {\bibinfo {author} {\bibfnamefont {Y.}~\bibnamefont
  {Aharonov}}, \bibinfo {author} {\bibfnamefont {L.}~\bibnamefont
  {Davidovich}},\ and\ \bibinfo {author} {\bibfnamefont {N.}~\bibnamefont
  {Zagury}},\ }\bibfield  {title} {\bibinfo {title} {Quantum random walks},\
  }\href@noop {} {\bibfield  {journal} {\bibinfo  {journal} {Physical Review
  A}\ }\textbf {\bibinfo {volume} {48}},\ \bibinfo {pages} {1687} (\bibinfo
  {year} {1993})}\BibitemShut {NoStop}%
\bibitem [{\citenamefont {Kempe}(2003)}]{kempe2003quantum}%
  \BibitemOpen
  \bibfield  {author} {\bibinfo {author} {\bibfnamefont {J.}~\bibnamefont
  {Kempe}},\ }\bibfield  {title} {\bibinfo {title} {Quantum random walks: an
  introductory overview},\ }\href@noop {} {\bibfield  {journal} {\bibinfo
  {journal} {Contemporary Physics}\ }\textbf {\bibinfo {volume} {44}},\
  \bibinfo {pages} {307} (\bibinfo {year} {2003})}\BibitemShut {NoStop}%
\bibitem [{\citenamefont {Aharonov}\ \emph {et~al.}(2001)\citenamefont
  {Aharonov}, \citenamefont {Ambainis}, \citenamefont {Kempe},\ and\
  \citenamefont {Vazirani}}]{aharonov2001quantum}%
  \BibitemOpen
  \bibfield  {author} {\bibinfo {author} {\bibfnamefont {D.}~\bibnamefont
  {Aharonov}}, \bibinfo {author} {\bibfnamefont {A.}~\bibnamefont {Ambainis}},
  \bibinfo {author} {\bibfnamefont {J.}~\bibnamefont {Kempe}},\ and\ \bibinfo
  {author} {\bibfnamefont {U.}~\bibnamefont {Vazirani}},\ }\bibfield  {title}
  {\bibinfo {title} {Quantum walks on graphs},\ }in\ \href@noop {} {\emph
  {\bibinfo {booktitle} {Proceedings of the thirty-third annual ACM symposium
  on Theory of computing}}}\ (\bibinfo {year} {2001})\ pp.\ \bibinfo {pages}
  {50--59}\BibitemShut {NoStop}%
\bibitem [{\citenamefont {Ambainis}\ \emph {et~al.}(2001)\citenamefont
  {Ambainis}, \citenamefont {Bach}, \citenamefont {Nayak}, \citenamefont
  {Vishwanath},\ and\ \citenamefont {Watrous}}]{ambainis2001one}%
  \BibitemOpen
  \bibfield  {author} {\bibinfo {author} {\bibfnamefont {A.}~\bibnamefont
  {Ambainis}}, \bibinfo {author} {\bibfnamefont {E.}~\bibnamefont {Bach}},
  \bibinfo {author} {\bibfnamefont {A.}~\bibnamefont {Nayak}}, \bibinfo
  {author} {\bibfnamefont {A.}~\bibnamefont {Vishwanath}},\ and\ \bibinfo
  {author} {\bibfnamefont {J.}~\bibnamefont {Watrous}},\ }\bibfield  {title}
  {\bibinfo {title} {One-dimensional quantum walks},\ }in\ \href@noop {} {\emph
  {\bibinfo {booktitle} {Proceedings of the thirty-third annual ACM symposium
  on Theory of computing}}}\ (\bibinfo {year} {2001})\ pp.\ \bibinfo {pages}
  {37--49}\BibitemShut {NoStop}%
\bibitem [{\citenamefont {Moore}\ and\ \citenamefont
  {Russell}(2002)}]{moore2002quantum}%
  \BibitemOpen
  \bibfield  {author} {\bibinfo {author} {\bibfnamefont {C.}~\bibnamefont
  {Moore}}\ and\ \bibinfo {author} {\bibfnamefont {A.}~\bibnamefont
  {Russell}},\ }\bibfield  {title} {\bibinfo {title} {Quantum walks on the
  hypercube},\ }in\ \href@noop {} {\emph {\bibinfo {booktitle} {International
  Workshop on Randomization and Approximation Techniques in Computer
  Science}}}\ (\bibinfo {organization} {Springer},\ \bibinfo {year} {2002})\
  pp.\ \bibinfo {pages} {164--178}\BibitemShut {NoStop}%
\bibitem [{\citenamefont {Douglas}\ and\ \citenamefont
  {Wang}(2009)}]{douglas2009efficient}%
  \BibitemOpen
  \bibfield  {author} {\bibinfo {author} {\bibfnamefont {B.}~\bibnamefont
  {Douglas}}\ and\ \bibinfo {author} {\bibfnamefont {J.}~\bibnamefont {Wang}},\
  }\bibfield  {title} {\bibinfo {title} {Efficient quantum circuit
  implementation of quantum walks},\ }\href@noop {} {\bibfield  {journal}
  {\bibinfo  {journal} {Physical Review A}\ }\textbf {\bibinfo {volume} {79}},\
  \bibinfo {pages} {052335} (\bibinfo {year} {2009})}\BibitemShut {NoStop}%
\bibitem [{\citenamefont {Grover}(1996)}]{grover1996fast}%
  \BibitemOpen
  \bibfield  {author} {\bibinfo {author} {\bibfnamefont {L.~K.}\ \bibnamefont
  {Grover}},\ }\bibfield  {title} {\bibinfo {title} {A fast quantum mechanical
  algorithm for database search},\ }in\ \href@noop {} {\emph {\bibinfo
  {booktitle} {Proceedings of the twenty-eighth annual ACM symposium on Theory
  of computing}}}\ (\bibinfo {year} {1996})\ pp.\ \bibinfo {pages}
  {212--219}\BibitemShut {NoStop}%
\bibitem [{\citenamefont {Lavor}\ \emph {et~al.}(2003)\citenamefont {Lavor},
  \citenamefont {Manssur},\ and\ \citenamefont {Portugal}}]{lavor2003grover}%
  \BibitemOpen
  \bibfield  {author} {\bibinfo {author} {\bibfnamefont {C.}~\bibnamefont
  {Lavor}}, \bibinfo {author} {\bibfnamefont {L.}~\bibnamefont {Manssur}},\
  and\ \bibinfo {author} {\bibfnamefont {R.}~\bibnamefont {Portugal}},\
  }\bibfield  {title} {\bibinfo {title} {Grover's algorithm: Quantum database
  search},\ }\href@noop {} {\bibfield  {journal} {\bibinfo  {journal} {arXiv
  preprint quant-ph/0301079}\ } (\bibinfo {year} {2003})}\BibitemShut {NoStop}%
\end{thebibliography}%

\appendix
\section{Circuit implementation of the SUM gates (the Modulo adder)}
\label{modadder}
This appendix describes the construction of a quantum circuit for a $d$-dimensional SUM gate. The inputs of this gate are the pair of $k$ qubits $A$ and $B$ corresponding to the pair of $d$-dimensional qudits, and the output is the modulo with respect to $d$ of the sum of the two inputs.

As described in Section \ref{sec_QRS}, the circuit consists of an RCA part and a modulo part. The RCA part adds the numbers of the same place in inputs $A$ and $B$ and stores the overflow of the place in the ``carry'' qubits for use in calculating the next place. Fig.\ref{mod5_num} shows what information is stored in which qubit during the calculation of the RCA part for $d=5$.
\begin{figure}[htbp]
    \begin{center}
        \includegraphics[width=8cm, keepaspectratio]{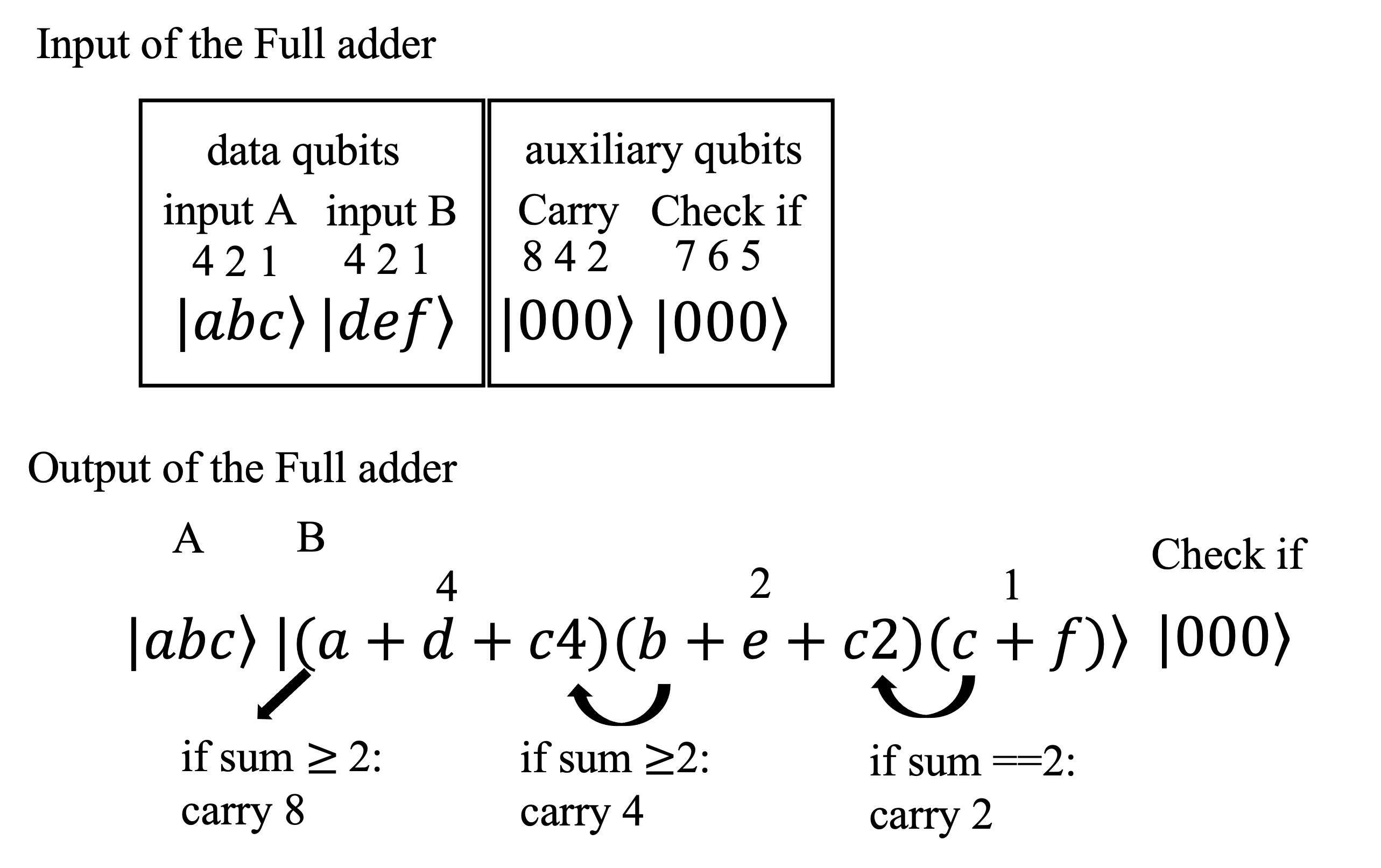}
        \caption{Input and output of the RCA operation. Carry qubits are used to store the overflow of each places for $4$, $2$, and $1$. Check-if qubits are not used here.}
        \label{mod5_num}
    \end{center}
\end{figure}

Next the modulo circuit converts the $\ket{(a+b)\,[\textrm{mod} \,2^k]}$ state stored in $B$ to the $\ket{(a+b)\,[\textrm{mod}\,d]}$ state. Such a conversion is necessary for the case in which the sum of the inputs is between $d$ to $2(d-1)$. We prepare a ``check if'' qubit for each of the numbers from $d$ to $2(d-1)$, and then the ``check if'' qubit is flagged (flipped from $\ket{0}$ to $\ket{1}$) when the sum of the inputs in each corresponding values $d$ to $2(d-1)$. In order to perform such an operation, we perform $C_{k}X$ gate with the data qubits $B$ as the control qubits with the corresponding ``check if'' qubit as the target qubit. At this time, we choose control and $0$-control based on the binary representation of each number. For normal controls, the gate is applied to the target when the control qubit is $\ket{1}$, but for $0$-control, the gate is applied when the control is $\ket{1}$. If one wants to invert the ``check if'' qubit for $5\,(101)$, use the control, $0$-control, and control for $C_{k}X$. Next, when the sum cannot be a number greater than $2^k$, the most significant carry can be used as the ``check if'' qubit of $2^k$. In this case, the $C_{k}X$ gate for ``check if'' qubit inversion is also unnecessary.

Consider the following example. If $d=5$, then a conversion is required when the sum of the inputs is between $5$ and $8$. Since $0\,[ \textrm{mod}\,8] = 0\,[\textrm{mod}\,5]$,... $4\,[\textrm{mod}\,8] = 4\,[\textrm{mod}\,5]$, no conversion is required for $d=0$ to $4$. You can also substitute $Carry_8$ as a ``check if'' qubit of $8$ in this case. After the flip of the ``check if'' qubits, converts the data qubits $B$ based on the ``check if'' qubit. As an example, we execute controlled $XIX, XXX, XIX, IXX$ gates with ``check if'' qubit as control qubit for the conversion from $5(101)$, $6(110)$, $7(111)$, $8(000)$ to $0(000)$, $1(001)$, $2(010)$, $3(011)$ for $d=5$.

It is important to note that when the maximum value that the sum can take $2(d-1)$ is greater than $2^k$, the largest ``carry'' qubit must be used as a control qubit at the same time. For example, in the $d=7$ case, the $mod\,2^k\rightarrow mod\,d$ transformation is needed when the sum is between $7$ and $14$.
However, to determine that the sum is $8,9,10,11,12$, not only must the data qubit be $000,001,010,011,100$, but $Carry_8$ must be $1$. This makes it necessary to use the $C_{k+1}X$ gate instead of $C_{k}X$. Such $C_{k+1}X$ gates can only be decomposed up to $\rm{C_2X}$ gates when multiplexing decomposition is used since the control qubits are in two qudits (data and carry). For the decomposition of $\rm{C_2X}$ gates to $\rm{CX}$ gates, we need to use the general method.

The case of constructing a modulo adder for the $5$ and $7$ dimensional qudit using $2^3$ multiplexed photon is shown in the Fig.\ref{mod5} and \ref{mod7}. 

\begin{figure}[htbp]
    \begin{center}
        \includegraphics[width=8cm, keepaspectratio]{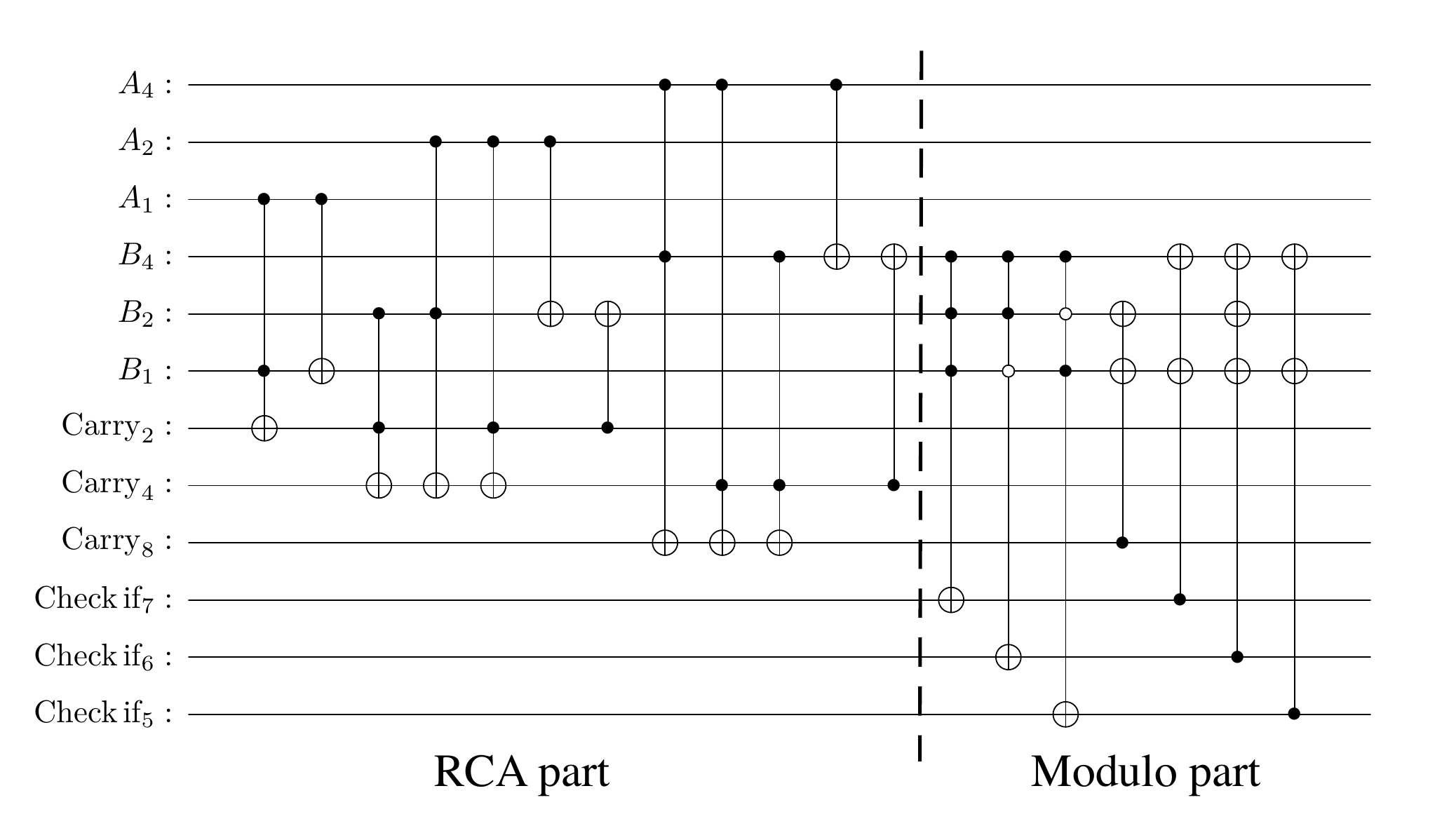}
        \caption{Circuit layout for the SUM gate for qudit of dimension-$5$ in the qubit system. It consists of a RCA and a modulo part. The black dots represent control qubits, and the white dots represent $0$-control qubits.}
        \label{mod5}
    \end{center}
\end{figure}

\begin{figure}[htbp]
    \begin{center}
        \includegraphics[width=8cm, keepaspectratio]{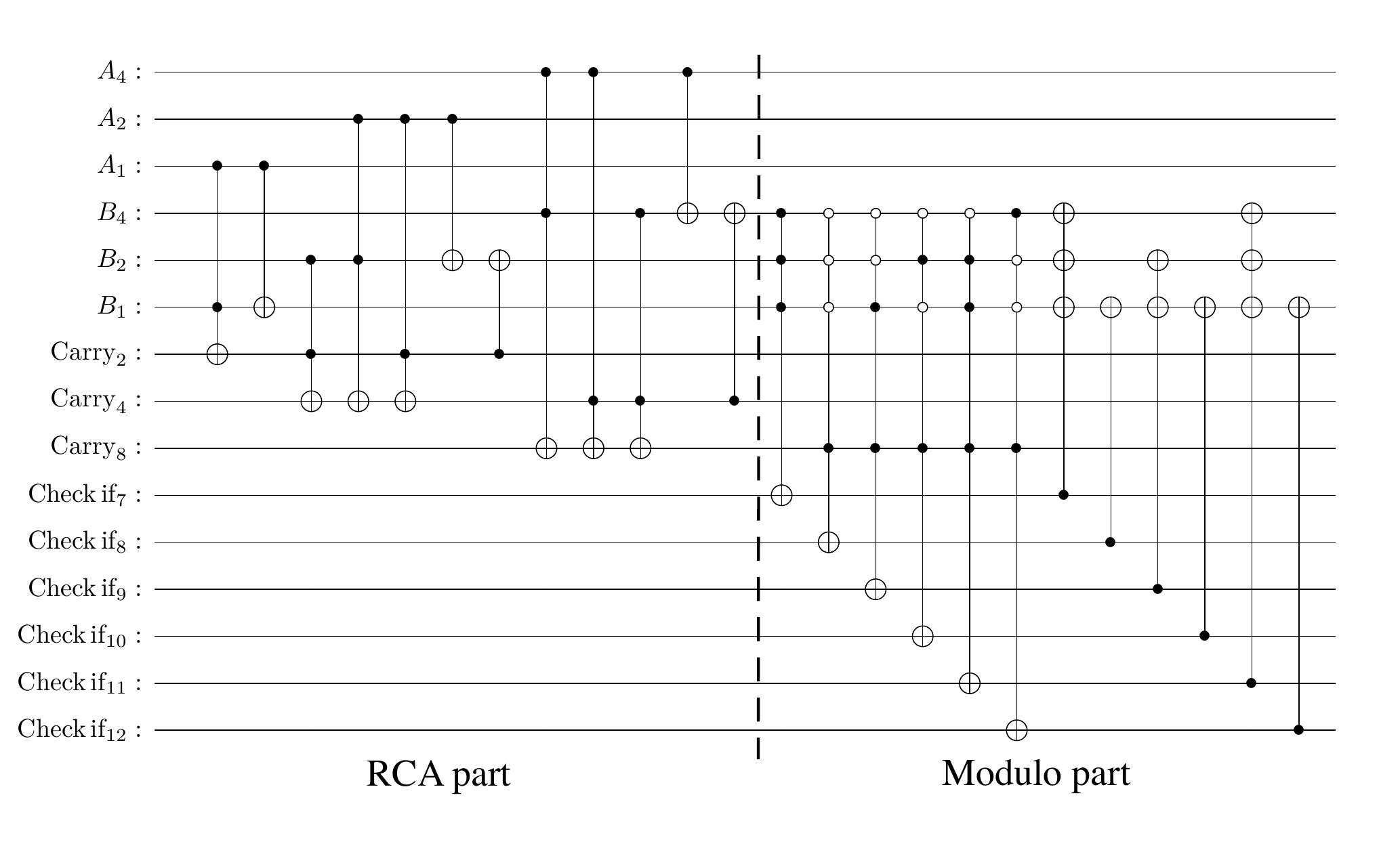}
        \caption{Circuit layout for the SUM gate for qudit of dimension-$7$ in the qubit system. This circuit includes the $C_{K+1}X = C_4X$ gates which has a control qubit in $\textrm{Carry}_8$.}
        \label{mod7}
    \end{center}
\end{figure}

\section{The encoding circuit for the $GF(2^m)$ QRS codes}
\label{2m}
Let us consider the $[[d-1,1,D\geqq \frac{N+1}{2}]]_d$ QRS code on $GF(2^m)$ with $d=2^m$. This code construction is based on $[[kN,k(N-2K),d\geqq K+1]]$ (written in binary form) for $N-2K=1$. To illustrate how this works let us take the $GF(2^2)$ code as a simple example. The Galois Field for the code is represented as
\begin{equation}
    GF(2^2) = \{0,1,\alpha, \alpha^2\}
\end{equation}
where we choose $x^2+x+1 = 0$ as the primitive polynomial. Then the elements of the field be written in polynomial and vector representation:
\begin{eqnarray*}
    GF(4) &=& \{0,1,\alpha, \alpha^2\}   \,\,\textrm{exponential representation}\\
    &=& \{0,1,\alpha,\alpha+1\} \,\textrm{polynomial representation}\\
    &=& \{00, 01, 10, 11\} \,\textrm{vector representation}
\end{eqnarray*}
We use the classical code $\mathcal{C}=[3,2,2]_4$ and it's dual $\mathcal{C}^\perp$ for our CSS construction. The generator polynomial for the code is given as
\begin{equation}
    g^\perp(x) = (x-1)(x-\alpha) = x^2 + (1+\alpha)x + \alpha
\end{equation}
from which we get the generator and parity check matrices:
\begin{equation}
    G = 
    \begin{pmatrix}
    1 & 0 & \alpha \\
    0 & 1 & \alpha^2
    \end{pmatrix}
\end{equation}
\begin{equation}
    H = 
    \begin{pmatrix}
    \alpha & \alpha^2 & 1 
    \end{pmatrix}.
\end{equation}
The quantum circuit for the encoder is shown in Fig.\ref{gf4_encoder}

Now we denote the qudit gates required for our implementation as $C1$, $C\alpha$, and $C\alpha^2$ defining them as follows:
\begin{eqnarray}
    C1(\ket{a}\ket{b}) &=&   \ket{a}\ket{a+b}\\
    C\alpha (\ket{a}\ket{b}) &=& \ket{a}\ket{\alpha a + b}\\
    C\alpha^2 (\ket{a}\ket{b}) &=& \ket{a}\ket{\alpha^2 a + b}
\end{eqnarray}
These gates correspond to the addition and multiplication of Galois elements for the calculation of coefficients of polynomials in Quantum Reed-Solomon codes. They can be implemented in a multiplexing system as shown in the Fig.\ref{gf4_gates}. As such, the gates of $\rm{GF} (2^m)$ can be constructed using only $\rm{CX}$ gates.

By generalizing the above method, we can calculate the cost of the gates we will need for implementing the $[[d-1,1,D\geqq \frac{N+1}{2}]]_d$ QRS code. First the $C1$ gate can be achieved with $k$ $\rm{CX}$ gates while the $C\alpha^n$ gates can be implemented with $\displaystyle \sum_{p=0}^{k-1}\rm{H_w}(\alpha^{n+p})$ $\rm{CX}$ gates for the $GF(2^k)$ system where $\rm{H_w}()$ is a function whose input is an exponential representation of a element of the Galois Field and whose output is a Hamming weight of the vector representation of the element. In this case, we do not need the $C_kX$ gate to implement the encoding circuit of the code. Therefore, the application of quantum multiplexing does not lead to any advantage in terms of the reduction of the number of gates.

\begin{figure}[htbp]
    \begin{center}
        \includegraphics[width=8cm, keepaspectratio]{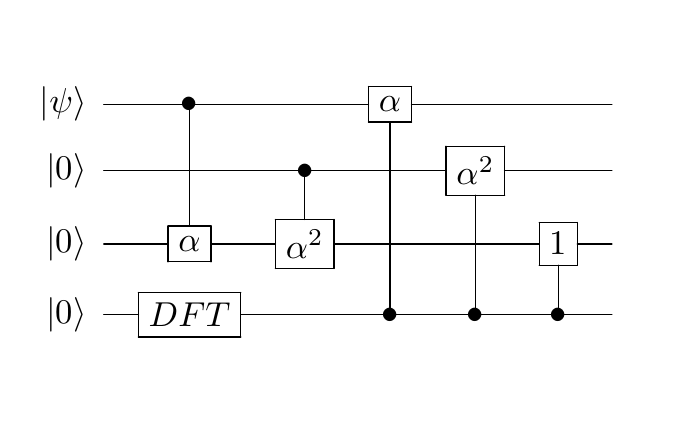}
        \caption{The encoding circuit of $[[3,1,D\geqq 2]]_4$ Quantum Reed-Solomon code. Each horizontal line represents a 4-dimensional qudit (ququart).  The two qu-dit gates defined by the black dot representing the control and the target with a square box containing the ``$1$'' ($\alpha, \alpha^2$) symbol represent the $C1$ ($C\alpha, C\alpha^2$) gates defined in (B5-7).}
        \label{gf4_encoder}
    \end{center}
\end{figure}
\begin{figure}[htbp]
    \begin{center}
        \includegraphics[width=9cm, keepaspectratio]{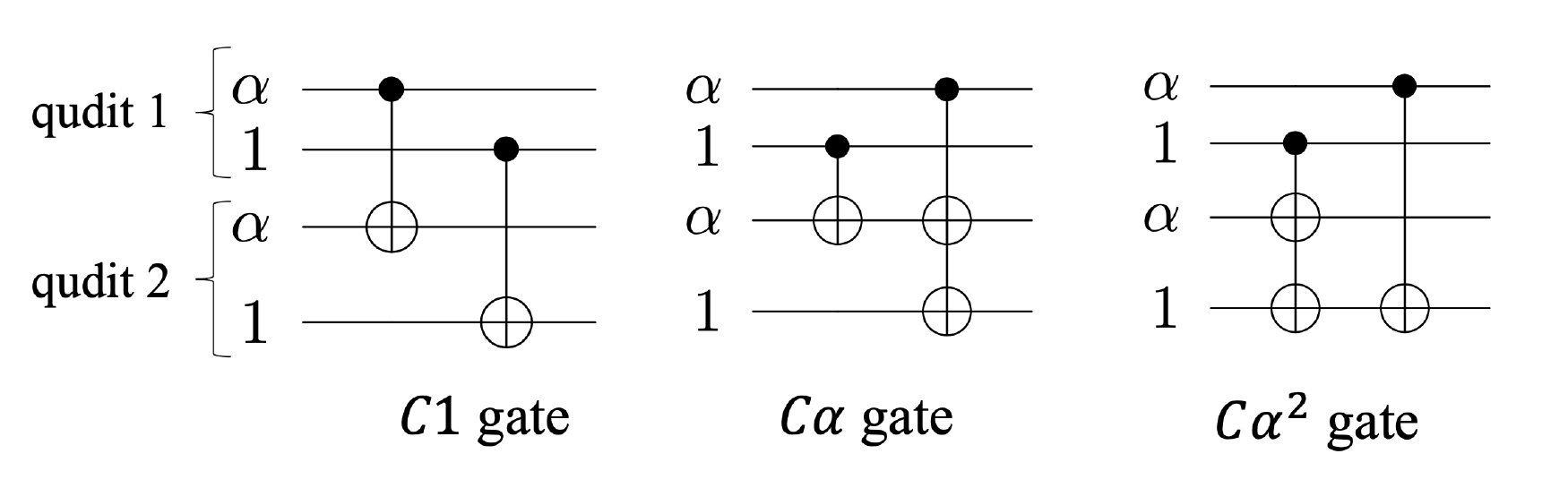}
        \caption{Qubit implementation of the $C1$, $C\alpha$, and $C\alpha^2$ gates.}
        \label{gf4_gates}
    \end{center}
\end{figure}
\newpage
\section{Proof for the multiplexed decomposition}
\label{Proof}
\begin{thm*}
A $C_{k}X$ gate, which has $k \in \ \mathbb{Z}^+$ control time-bin qubits in a photon and a target qubit in another photon, can be decomposed into a single $\rm{CX}$ gate alongside several optical switches.
\end{thm*}
\begin{proof}
Let $S(k)$ be the statement that the $C_{k}X$ gate in the theorem can be decomposed into one $\rm{CX}$ gate and several optical switches. 

We will now give a proof by induction on $k$ beginning with $k=1$. Since $C_{1}X$ gate is a $\rm{CX}$ gate from a control photon to a target photon, it holds by definition. Now for any integer $k\geqq 1$, if $S(k)$ holds, then $S(k+1)$ also holds. Assume the induction hypothesis $S(k)$ is true. Since the $C_{k+1}X$ gate can be realized by controlling the $C_{k}X$ gate with the ($k+1$)-th timebin qubit, it can be decomposed as shown in Fig.\ref{dec_ck+1xproof}. From $S(k)$, the $C_{k}X$ gate can be implemented with one $\rm{CX}$ and several OSs, therefore $S(k+1)$ holds.

\begin{figure}[h!]
    \subfigure[]{%
        \includegraphics[clip, width=0.44\columnwidth]{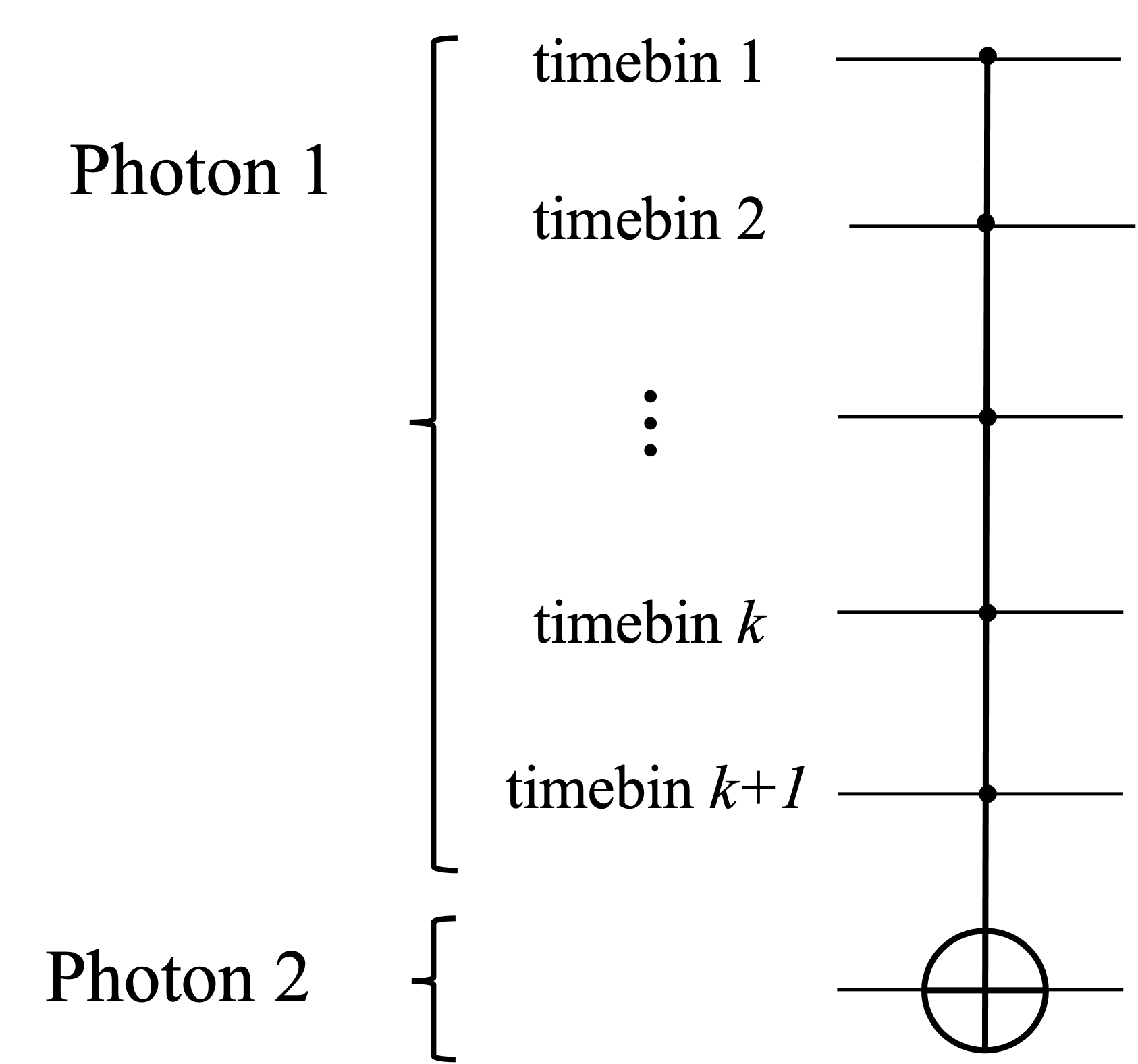}}%
    \subfigure[]{%
        \includegraphics[clip, width=0.44\columnwidth]{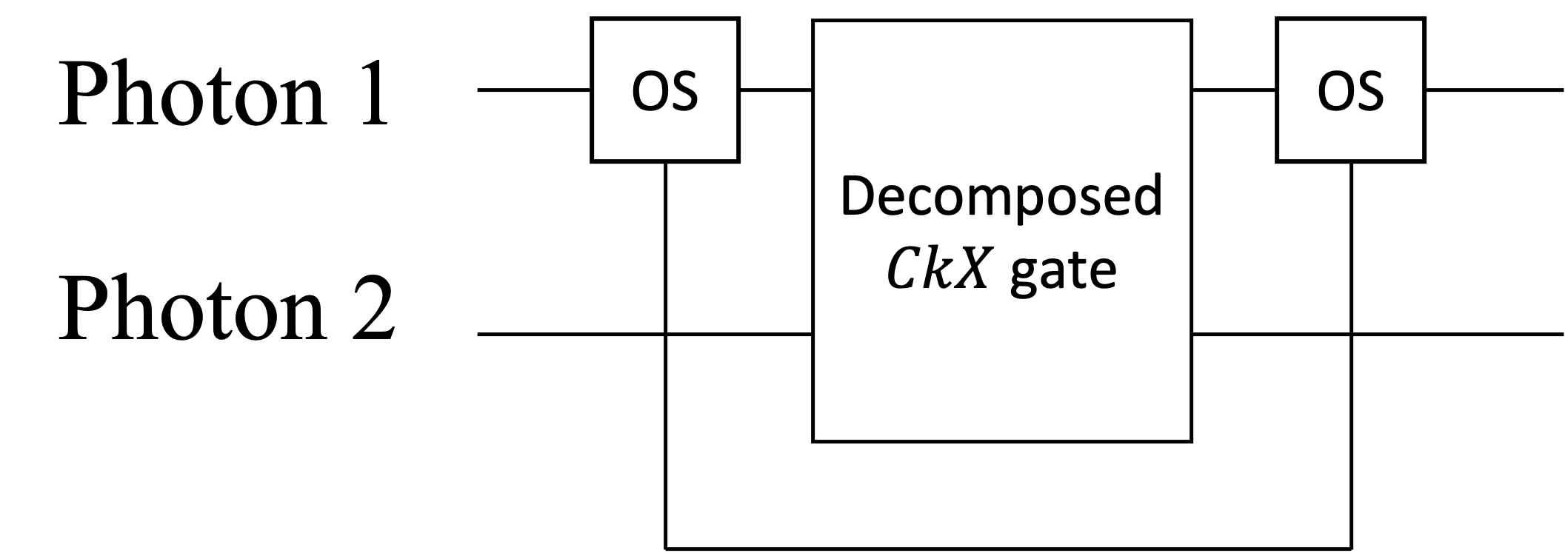}}%
    \caption{(a) $C_{k+1}X$ between two photons. (b) The circuit implementation for (a). Optical switches divides the (k+1)-th timebin qubit.}
    \label{dec_ck+1xproof}
\end{figure}
\end{proof}
Note that using the implementation of $\rm{CX}$ gate in the Appendix C of \cite{piparo2020aggregating}, the target qubit can be either polarization or timebin.

\end{document}